\documentclass[apalike]{article}
\usepackage[authoryear]{natbib} 
\usepackage{apalike}
\usepackage[dvips]{graphicx}
\usepackage{amsmath}
\usepackage{amsfonts}
\usepackage{amssymb}
\usepackage{amsthm}

\usepackage{multirow}
\usepackage{array}
\usepackage{tabularx}
\usepackage{mathrsfs}
\usepackage{url}

\newtheorem{theorem}{Theorem}
\newtheorem{assumption}{Assumption}

\newtheorem{proposition}{Proposition}

\makeatletter

\@addtoreset{equation}{section}
\makeatother

\begin{document}

\markboth{Yukihiro Tsuzuki}{Boundary conditions at infinity for Black--Scholes equations}
\title{Boundary conditions at infinity for Black--Scholes equations\thanks{We would like to thank Professor Yoshiki Otobe for the helpful discussions.
This study was supported by JSPS KAKENHI, Grant Number 23K01466.
}}
\author{Yukihiro Tsuzuki\thanks{Faculty of Economics and Law, Shinshu University, 3-1-1 Asahi, Matsumoto, Nagano 390-8621, Japan; yukihirotsuzuki@shinshu-u.ac.jp}}

\maketitle

\abstract{
We propose a numerical procedure for computing the prices of European options, in which the underlying asset price is a Markovian strict local martingale.
If the underlying process is a strict local martingale and the payoff is of linear growth, multiple solutions exist for the corresponding Black--Scholes equations.
When numerical schemes such as finite difference methods are applied,
a boundary condition at infinity must be specified, which determines a solution among the candidates.
The minimal solution, which is considered as the derivative price, is obtained by our boundary condition.
The stability of our procedure is supported by the fact that our numerical solution satisfies a discrete maximum principle. In addition,
its accuracy is demonstrated through numerical experiments in comparison with the methods proposed in the literature.
}



\clearpage
\section{Introduction}
\label{intro}
In this study, we propose a numerical procedure for computing the prices of European options in the presence of a bubble.
A bubble is defined as a discounted price process that is modeled by a strict local martingale (i.e., a local martingale but not a true martingale) under a risk-neutral measure.
Bubbles have been studied extensively in mathematical finance (for example, in \citet{CoxHobson2005},
\citet{HestonLoewensteinGregory2006},
\citet{JarrowProtterShimbo2007},
\citet{JarrowProtterShimbo2010},
\citet{PalProtter2010},
\citet{CarrFisherRuf2014},
and the references therein).
Many standard results in option pricing theory may fail in such markets.
For example, put-call parity does not hold,
the European call option price is neither convex in the underlying asset price nor increasing in maturity,
the price of an American call exceeds that of a European call,
and solutions to the Black--Scholes equations are non-unique.

We focus on the issue of non-uniqueness of the Black--Scholes equations.
In general, derivative prices are solutions to the corresponding partial differential equations or Cauchy problems, which are referred to as Black--Scholes equations in the literature.
The existence and uniqueness of these solutions have been studied in detail by \citet{HeathSchweizer2000} and \citet{JansonTysk2006}.
However, multiple solutions may exist in the presence of a bubble.
More precisely, \citet{Tysk2009} show that the uniqueness holds if the payoff function is of strictly sublinear growth
and does not hold if it is of linear growth, such as those of call options and forward contracts.
The issue of non-uniqueness has been addressed from a theoretical perspective.
The derivative price is characterized as the smallest non-negative supersolution according to the fundamental theorem of asset pricing (\citet{Delbaen1994}),
and the smallest solution is characterized as a unique solution to an alternative Cauchy problem (Theorem 6.2 of \citet{Cetin2018}).
The smallest solution is provided by the expectation of the payoff, which is referred to as a {\it stochastic solution}, whereas a function that satisfies the Black--Scholes equation is a {\it classical solution}.
However, despite their practical importance, numerical procedures have not been extensively developed.
Specifying the spatial boundary conditions is one difficulty when applying numerical methods that are set up on a finite grid, such as finite difference methods.
A standard specification, i.e., specifying the boundary conditions as the payoff function, does not lead to the desired results in the presence of a bubble if the payoff function is of linear growth.
Therefore, the boundary conditions must be specified carefully.

Several methods have been proposed in the literature.
For instance, \citet{EkstromLotstedtSydowTysk2011} propose the Neumann boundary conditions; that is, the first-order derivative is $0$,
whereas \citet{SongYang2015} propose setting $0$ as the boundary condition and revising the terminal conditions such that the terminal-boundary datum is continuous.
The prices from these procedures converge to the true price.
However, these boundary conditions do not sufficiently consider the fact that 
the solution, which is of sublinear growth (Corollary 6.1 of \citet{Cetin2018}), approaches a payoff function, which is of linear growth, as the time-to-maturity tends to zero.
The application of finite schemes to the alternative Cauchy problem proposed by \citet{Cetin2018} is worth investigating. However, to the best of our knowledge, such an investigation has not been reported in the literature.
We confirm that a straightforward application leads to almost the same method as that in \citet{SongYang2015} without a payoff revision.
The change of variable that transforms infinity into zero makes this method different from the approach of \citet{SongYang2015}
and it is expected to reduce spatial truncation errors.
However, a singular coefficient can appear in the alternative partial differential equations.
This can cause the problem to be {\it convection dominated}, which is common
in the fields of computational fluid dynamics and pricing Asian options (see, for example, \citet{Zvan1997} and Chapter 10 of \citet{Duffy2006}).
In this case, numerical solutions show spurious oscillations or a very small mesh size must be selected.
In contrast to the methods above, \citet{tsuzuki2023pitmans} proposes non-constant boundary conditions using Pitman's theorem for the three-dimensional Bessel process.
This procedure is numerically investigated in this study,
from which further improvements are revealed to be necessary.

This study proposes boundary conditions for computing the prices of forward contracts and the other European options, respectively,
in which the Black--Scholes equation has a unique solution among an appropriate class of considered functions.
These conditions are based on the formula that expresses the derivative price at infinity as the integral of the payoff function
with respect to the measure that is obtained by the forward price function, which is concave in the presence of a bubble.
The case of the forward contract is typical and important.
In this case, whereas the formula holds with inequality for a classical solution; that is,
its value at infinity is larger than the integral that is computed using the solution,
the minimum of the former and maximum of the latter are attained by the stochastic solution,
and they are equal.
In the other cases, the boundary conditions are obtained as boundary values
that are computed by the stochastic solution of the forward contract.
We express the underlying process with a transient diffusion and use its time-reversal to derive the formula.
A condition on the volatility is required to use this technique.
This formula is also obtained as the limit of the boundary conditions of \citet{tsuzuki2023pitmans}.

Our numerical procedure has the following three advantages.
First, payoff functions are not restricted to those of at most linear growth, which is assumed for the methods proposed in the literature.
The second advantage is stability.
Our finite scheme satisfies a discrete maximum principle,
which ensures that the numerical solutions do not admit local maxima/minima.
Therefore, the solutions for non-decreasing payoff functions are free from oscillations.
In addition, our numerical forward prices remain increasing and concave with respect to the underlying prices and decreasing with respect to the time-to-maturity, which are expected to be the case.
Finally, as demonstrated theoretically and numerically, our procedure is more accurate than the others.
The discrete maximum principle also shows
that the numerical solutions of \citet{EkstromLotstedtSydowTysk2011} are less than ours.
Our method outperforms that of \citet{tsuzuki2023pitmans} by construction
because, as mentioned previously, our solution is obtained as the limit of the boundary conditions of \citet{tsuzuki2023pitmans}.
\citet{SongYang2015} produces much smaller results than those of \citet{tsuzuki2023pitmans} owing to smaller boundary values.
The numerical tests in this study confirm that our method is more accurate than the others except in one case, in which the results of \citet{Cetin2018} slightly outperform ours.

The following contributions are also made in this study:
In contrast to the case of a true martingale,
the forward price in the presence of a bubble is a concave function with respect to the underlying price,
which inspires us to investigate distributions that are implied in the forward prices.
We provide three interpretations of the forward prices in terms of probabilities, aided by the results of \citet{profeta2010}.
Whether the forward value at infinity is finite is an interesting question, and
to the best of our knowledge, whether it is always finite remains an open question.
\citet{Tysk2009} provide volatility growth conditions for the forward price at infinity to be finite.
We weaken this condition and find that the finiteness holds even in cases in which the prices are considered to increase slowly, but not necessarily to be bounded.

The problem of the non-uniqueness of the Cauchy problem is closely related to the problem of a relevant diffusion being a strict local martingale.
\citet{BayraktarXing2010} provide a necessary and sufficient condition for the uniqueness in terms of a diffusion coefficient.
\citet{CetinLarsen2023} study uniqueness in Cauchy problems for real-valued one-dimensional diffusions.
Cauchy problems with multiple solutions have emerged in mathematical finance and diffusion theory.
\citet{Fernholzkaratzas2010} consider the highest return on investment in an equity market model,
whereas \citet{KaratzasRuf2013} compute the distribution function of the time-to-explosion.
These solutions are characterized by the smallest non-negative solutions of partial differential equations (or inequalities) with multiple possible solutions.
The computation of the distribution function of the time-to-explosion can be reduced to the alternative problem of \citet{Cetin2018}, to which our numerical method can be applied.

A non-negative strict local martingale defines a measure that is absolutely continuous but not equivalent to the original measure via an $h$-transform.
Conversely, a strict local martingale is obtained as the Radon--Nikodym process from two probability measures that are not equivalent.
This correspondence is studied by \citet{PalProtter2010} for continuous local martingales
and is extended to the case of c{\'a}dl{\'a}g local martingales by \citet{KardarasKreherNikeghbali2015}.
The current study is carried out on the two probability measures that are related via an $h$-transform, as in \citet{PalProtter2010}.
We construct a regular transient diffusion on the dominated measure such that the local martingale is expressed as a transformation of this diffusion,
and use the duality relationship of the two measures to derive our boundary condition.

In the following section, we formulate the problem.
In Section \ref{sec:main},
we provide the main theorem of this study,
in which the derivative price at infinity is represented with its payoff and the forward prices.
In addition, we derive distributions that are implied in forward prices,
provide a condition for the boundedness of the forward prices,
and discuss the uniqueness.
Section \ref{sec:example} presents models that admit analytical expressions: the constant elasticity of variance and quadratic normal volatility models.
Finally, in Section \ref{sec:numerical},
we provide descriptions of our numerical schemes as well as those proposed in the literature,
investigate their stability,
and present numerical tests.

\section{Problem}
\label{sec:problem}
Let $\sigma : (0,\infty) \longrightarrow (0,\infty)$ be a locally H\"{o}lder continuous function with exponent $1/2$
and $Y$ be the solution to the stochastic differential equation
\begin{eqnarray}
Y_{t} = Y_{0} + \int_{0}^{t} \sigma(Y_{u}) d\beta_{u}, \label{eq:underlying}
\end{eqnarray}
where $Y_{0}>0$ and $\beta$ is a Brownian motion.
In addition, $y=0$ is assumed as an absorbing boundary.
That is, if $Y_{t}=0$ at some $t$, $Y$ remains at $0$ at all times after $t$.
The condition on $\sigma$ ensures a unique strong solution that is absorbed at $0$.

We interpret $Y$ as the underlying price process,
and assume that it is defined under a risk-neutral measure
and that the risk-free interest rate is zero.
In this market model, we consider the value of a contingent claim that pays $h(Y_{T})$ at time $T$ for a non-negative continuous function $h$.
As the underlying price process $Y$ is a local martingale under a risk-neutral measure,
the arbitrage pricing theory suggests that the value of the contingent claim with time-to-maturity $\tau \ge 0$ and the underlying price $y$ is 
$v^{h}(\tau,y) := E_{y}[h(Y_{\tau})]$, where $E_{y}[\cdot]$ is the expectation operator that is conditioned on $Y_{0}=y>0$,
and if $v^{h}$ is sufficiently differentiable, it satisfies the partial differential equation
\begin{eqnarray}
\label{eq:bse}
v_{\tau} = \frac{1}{2} \sigma^{2} v_{yy}
\end{eqnarray}
for $(\tau,y) \in (0,\infty) \times (0,\infty)$,
with the initial condition
\begin{eqnarray}
\label{eq:initial}
v(0,y) = h(y).
\end{eqnarray}
The function $v^{h}$ is referred to as a stochastic solution, whereas a function that satisfies (\ref{eq:bse}) is a classical solution.

A bubble is a discounted price process that is a non-negative strict local martingale, and hence, a supermartingale, under a risk-neutral measure.
In the presence of a bubble, the Black--Scholes equation (\ref{eq:bse}) with (\ref{eq:initial}) admits multiple solutions for $h$ of linear growth,
among which the stochastic solution $v^{h}$ is characterized as the smallest solution.
A typical and an important example is the case of the forward contract; that is, $h(y)=y$,
where we express the stochastic solution using $v^{*}$ instead of $v^{h}$ for special emphasis.
The stochastic solutions $v^{*}$ and $v(\tau,y)=y$
are two distinct solutions owing to the supermartingale property of $Y$:
\begin{eqnarray}
v^{*}(\tau,y) := E_{y}[Y_{\tau}] < y = v(\tau,y).
\end{eqnarray}
The stochastic solution $v^{*}(\tau,y)$ has very different properties from the solution $v(\tau,y)=y$; 
that is, $v^{*}(\tau,y)$ is decreasing in $\tau$ and increasing but concave in $y$.

When applying numerical methods that are set up on a finite grid, such as $(0,T) \times (0,n)$,
the Black--Scholes equation is reduced to a partial differential equation (\ref{eq:bse}) on $(0,T) \times (0,n)$ with (\ref{eq:initial}).
In this case, appropriate additional spatial boundary conditions at $y=n$, such that the equation in the finite domain $(0,T) \times (0,n)$ has a unique solution, are required.
The problem is reduced to specifying the boundary condition.
Several boundary conditions for the case of $h$ of linear growth have been proposed to date.
We provide a brief literature review in Section \ref{sec:procedures}.

To address the issue of non-uniqueness with the Black--Scholes equations,
\citet{Cetin2018} establishes a characterization of the stochastic solution in terms of the unique solution of an alternative Cauchy problem
\begin{eqnarray}
\label{eq:cetin}
\left\{ \begin{array}{l}
w_{\tau}(\tau,y) = \frac{1}{2}\sigma^{2}(y)w_{yy}+\frac{\sigma^{2}(y)}{y} w_{y}, \\
w(0,y) = \frac{h(y)}{y},
\end{array} \right.
\end{eqnarray}
where $\lim_{y\rightarrow 0} \sup_{u \le \tau} yw(u,y)<\infty$,
and $\lim_{n \rightarrow \infty} w(\tau_{n},y_{n})=0$ for any $y_{n}\nearrow \infty$ and $\tau_{n} \rightarrow \tau > 0$.
The stochastic solution to the original problem is $v^{h}(\tau,y)=yw(\tau,y)$.
Refer to Theorem 6.2 of \citet{Cetin2018} for further details.
The finite difference scheme for (\ref{eq:cetin}) on $(0,T) \times (0,n)$ with $w(\tau,n)=0$
is equivalent to that for the original problem with $v(\tau,n)=0$,
which is the same as the method of \citet{SongYang2015}.
Instead, we examine the finite scheme for (\ref{eq:cetin}) with a change of variable in Section \ref{sec:numerical}.

We mention a related problem.
\citet{KaratzasRuf2013} consider a diffusion that takes values in $(l,r)$ defined by the unique solution to
\begin{eqnarray}
X_{t} = x + \int_{0}^{t} s(X_{u}) (d\beta_{u}+b(X_{u})du),\; x \in (l,r),
\end{eqnarray}
where $l,r$ are constant, and $s$ and $b$ are Borel functions with appropriate conditions.
The complementary distribution function of the time-to-explosion of $X$;
that is, $S = \lim_{n \rightarrow \infty} \inf \{t : X_{t} \notin (l+1/n,r-1/n) \}$, is the unique solution to
\begin{eqnarray}
\left\{ \begin{array}{l}
u_{\tau} = \frac{1}{2}s^{2}u_{xx} + bs u_{x}, \\
u(0,x) = 1,\\
u(\tau,0)=0.
\end{array} \right.
\end{eqnarray}
The solution $u(\tau,x)$ to this problem is provided by the solution $w(\tau,y)$ to (\ref{eq:cetin}) with $h(y)=y$
by an appropriate change of variable $y=\varphi(x)$,
where $\varphi$ is determined by $\varphi^{\prime\prime}/\varphi^{\prime}-2\varphi^{\prime}/\varphi=-2b/s$.
In this case, if $Y$, which is determined by (\ref{eq:underlying}) with $\sigma(y) = s(\varphi^{-1}(y))$, is a strict local martingale, 
there exists a solution $u$, or equivalently $w$, that is not constant.
A difficulty is that the problem may be convection dominated, as noted in Section \ref{sec:cetinstability}.
In such a case, our numerical procedure is more promising.
It can be applied to the problem formally,
although $\sigma$ must satisfy a specific condition; that is, Assumption \ref{ass:ass_vol} below.

\section{Boundary conditions at infinity}
\label{sec:main}
In this section, we present the main theorem.
Hereafter in this study, we impose the following assumption on $\sigma$:
\begin{assumption}
\label{ass:ass_vol}
The volatility function $\sigma$ is of class $C^{1}((0,\infty))$
and satisfies 
\begin{itemize}
\item[(a)] $\int_{0}^{1} \frac{d\eta}{\sigma(\eta)} = \infty$ and $\int_{1}^{\infty} \frac{d\eta}{\sigma(\eta)} < \infty$;
\item[(b)] $Y$ is strictly positive: $\lim_{y \downarrow 0} \int_{y}^{1}\frac{\eta-y}{\sigma(\eta)^{2}} d\eta = \infty$; and 
\item[(c)] $Y$ is a strict local martingale: $\int_{1}^{\infty} \frac{\eta}{\sigma(\eta)^{2}} d\eta < \infty$.
\end{itemize}
\end{assumption}

We use (a) to express the law of $Y$ using a transient diffusion in Section \ref{sec:duality}.
Condition (b) ensures that $Y$ is strictly positive according to Feller's test (see, for example, Propositions 5.22 and 5.29 in Chapter 5 of \citet{karatzas}),
and (c) is necessary and sufficient for $Y$ to be a strict local martingale according to \citet{Kotani2006}.

The main theorem of this study is as follows:
\begin{theorem}
\label{thm:main}
Suppose that $\sigma$ satisfies Assumption \ref{ass:ass_vol}.
Then, for a non-negative, continuous function $h$,
we obtain
\begin{eqnarray}
v^{h}(\tau,\infty)
= -2 \int_{0}^{\infty} v_{\tau}^{*}(\tau,y) \frac{h(y)}{\sigma(y)^{2}} dy
= -\int_{0}^{\infty} h(y)v_{yy}^{*}(\tau,y)dy, \label{eq:boundary}
\end{eqnarray}
for all $\tau > 0$.
\end{theorem}

We provide several remarks regarding the theorem.

The stochastic solution $v^{*}$ is a classical solution according to Theorem 3.2 of \citet{Tysk2009} because the initial condition $h$ is of at most linear growth.
See Theorem 4.2 of \citet{ruf2012} for the differentiability of $v^{h}$ if $h$ is of more than linear growth.
However, we do not require $v^{h}$ to be sufficiently differentiable.

The application of the integration-by-parts formula to a smooth function $v(\tau,\cdot)$ yields
\begin{eqnarray}
v(\tau,y)
&=&
v(\tau,0)-\lim_{\eta\rightarrow 0} \eta v_{y}(\tau,\eta)+ yv_{y}(\tau,y) - \int_{0}^{y} \eta v_{yy}(\tau,\eta) d\eta \label{eq:ibp1}\\
&=&
v(\tau,0)-\lim_{\eta \rightarrow 0} \eta v_{y}(\tau,\eta)+yv_{y}(\tau,\infty)
-\int_{0}^{\infty} \left(y \wedge \eta \right)v_{yy}(\tau,\eta) d\eta. \label{eq:ibp2}
\end{eqnarray}
The boundary condition (\ref{eq:boundary}) for $h(y)=y$ can be obtained by letting $y \longrightarrow \infty$
under the conditions
$v(\tau,0) = \lim_{\eta \rightarrow 0} \eta v_{y}(\tau,\eta)=0$, which can be supposed,
and 
$\lim_{y \rightarrow \infty}yv_{y}(\tau,y)=0$ or $v_{y}(\tau,\infty)=0$.
However, we do not use this argument to prove Theorem \ref{thm:main}.
The above conditions at infinity are consequences rather than assumptions.
Similarly, for an initial condition $h$ that is not identical to $y$,
we obtain an alternative expression for $v^{h}(\tau,\infty)$ under appropriate assumptions:
\begin{eqnarray}
v^{h}(\tau,\infty)
= -\lim_{y \rightarrow \infty} \int_{0}^{y} \eta v_{yy}^{h}(\tau,\eta) d\eta. \label{eq:boundary2}
\end{eqnarray}
We note that the integrand of the integral of the right-hand side is not necessarily integrable in $[0,\infty)$.
Although we could use this alternative expression as a boundary condition,
we recommend the use of (\ref{eq:boundary}), especially when $h$ exhibits more than linear growth.
See Section \ref{sec:procedures} for further details.

The proof of Theorem \ref{thm:main} is presented in Section \ref{sec:duality}.
Interpretations of the forward price $v^{*}(\tau,\cdot)$ in terms of probabilities are provided in Section \ref{sec:as_prob},
a sufficient condition for $v^{*}(\tau,\infty) < \infty$ for $\tau > 0$
is derived in Section \ref{sec:finiteness},
and conditions for the uniqueness are provided in Section \ref{sec:uniqueness}.

\subsection{Proof of Theorem \ref{thm:main}}
\label{sec:duality}
In this section, we prove the boundary condition (\ref{eq:boundary})
using time-reversal for a regular transient diffusion $X$ on $(0,\infty)$ such that
the law of $Y$ is the same as $f(X)$,
where $f$ is defined by
\begin{eqnarray}
f^{-1} = \int_{\cdot}^{\infty} \frac{d\eta}{\sigma(\eta)}. \label{eq:def_f}
\end{eqnarray}
The construction of $f(X)$ is performed in the same manner as Proposition 6 of \citet{PalProtter2010},
which states that any strict local martingale that is strictly positive can be obtained as the reciprocal of a martingale under an $h$-transform.

We construct the transient diffusion $X$ on the canonical sample space $\Omega:=C([0,\infty))$,
where $C([0,\infty))$ denotes the space of continuous paths $\omega:[0,\infty) \longrightarrow [0,\infty]$ and $0,+\infty$ are assumed to be absorbing points.
In this space, $X$ denotes the coordinate process; that is,
$X_{t}(\omega):=\omega(t)$ for $\omega\in \Omega$.
We introduce the sigma algebras $\mathcal{F}_{t}^{0} := \bigcap_{\varepsilon > 0} \sigma(X_{u}; u \le t+\varepsilon)$.
We consider probabilities on $(\Omega,\mathcal{F}_{\infty}^{0})$
such that $X$ is a $[0,\infty]$-valued regular diffusion.
For such a probability measure $P$, the expectation operator is also denoted by $P$.

We can construct two probability measures $P_{x}^{f}$ and $P_{x}^{1/f}$ for $x>0$
on the canonical space such that $f(X)$ has the same law as $Y$ under $P_{x}^{f}$,
and $f(X_{0})/f(X)$ is a non-negative martingale with respect to $P_{x}^{1/f}$,
which is the Radon--Nikodym derivative of $dP_{x}^{f}/dP_{x}^{1/f}$.
More precisely, first,
let $P_{x}^{f}$ be the law of the solution to the stochastic differential equation
\begin{eqnarray}
X_{t} = \beta_{t} - \int_{0}^{t} \frac{1}{2}T_{f}(X_{u}) du, \label{eq:latant}
\end{eqnarray}
where
$\beta$ is a Brownian motion starting at $x$
and $T_{f}:=f^{\prime\prime}/f^{\prime}$.
Second,
let $P_{x}^{1/f}$ be the law of the solution to the stochastic differential equation
\begin{eqnarray}
X_{t} = \beta_{t}^{1/f} - \int_{0}^{t} \frac{1}{2}T_{1/f}(X_{u}) du,\label{eq:reciprocal_latant}
\end{eqnarray}
where $\beta^{1/f}$ is a Brownian motion starting at $x$.
Then, according to Ito's formula,
the law of $f(X)$ with respect to $P_{x}^{f}$ is the same as that of $Y$,
and $1/f(X)$ is a local martingale with respect to $P_{x}^{1/f}$.
In particular, (b) of Assumption \ref{ass:ass_vol} yields $P_{x}^{f}[\inf_{t \ge 0} X_{t} > 0, \lim_{t \rightarrow \infty} X_{t} = \infty ] = 1$;
that is, $X$ is a strictly positive transient diffusion.
Furthermore, we can relate $P_{x}^{f}$ and $P_{x}^{1/f}$ such that 
\begin{eqnarray}
P_{x}^{f} = \frac{1/f(X_{T_{n}})}{1/f(X_{0})}P_{x}^{1/f} \label{eq:equivalence}
\end{eqnarray}
holds on $\mathcal{F}_{T_{n}}^{0}$ for all $n > 0$,
where $T_{n} := \inf\{t>0: f(X_{t}) \notin (1/n,n) \}$.
By letting $n \longrightarrow \infty$ in (\ref{eq:equivalence}),
we obtain the absolutely continuous relationship
\begin{eqnarray}
P_{x}^{f} \left[\frac{f(X_{\tau})}{f(X_{0})}: \Gamma \right]
= P_{x}^{1/f}[\Gamma : \tau < T_{0}]
\end{eqnarray}
for any $\Gamma \in \mathcal{F}_{\tau}^{0}$ and $\tau > 0$,
where $T_{0} := \inf\{t>0:X_{t}=0 \}$.
In particular,
\begin{eqnarray}
P_{x}^{1/f} \left[\frac{f(X_{0})}{f(X_{\tau})}\right]
=
P_{x}^{1/f} \left[\frac{f(X_{0})}{f(X_{\tau})}:\tau < T_{0}\right]
=1
\end{eqnarray}
shows that $1/f(X)$ is a true martingale with respect to $P_{x}^{1/f}$.

The regular diffusion $(X,P_{x}^{f})_{x>0}$ on $(0,\infty)$ extends to a continuous Feller process on $[0,\infty)$ such that $P_{0}^{f}[X_{t}>0]=1$ for $t > 0$ (Theorem 33.13 of \citet{Kallenberg2021}),
because
the presence of a bubble implies that
$0$ is an entrance boundary for $(X,P_{x}^{f})_{x>0}$:
\begin{eqnarray}
\frac{1}{2} \int_{0}^{f^{-1}(1)} f(\xi)m_{f}(d\xi)
=\int_{1}^{\infty} \frac{\eta}{\sigma(\eta)^{2}} d\eta
< \infty,
\end{eqnarray}
where $m_{f}(dx)=2\frac{dx}{|f^{\prime}(x)|}$ is the speed measure of $X$ under $P_{x}^{f}$. 
We denote the law of $X$ starting from $0$ as $P_{0}^{f}$.

The boundary condition (\ref{eq:boundary}) follows from the following proposition, together with the expressions $v^{h}(\tau,\infty)=P_{0}^{f}[h(f(X_{\tau}))]$
and $v^{*}(\tau,y)=yP_{f^{-1}(y)}^{1/f}[\tau < T_{0}]$.
\begin{proposition}
\label{prop:doubleintegral}
For $\rho,\lambda \in C_{K}((0,\infty))$, we obtain
\begin{eqnarray}
\int_{0}^{\infty} \rho(t)P_{0}^{f}[\lambda(X_{t})h(f(X_{t}))] dt
= \int_{0}^{\infty} \lambda(x)P_{x}^{1/f} [\rho(T_{0})] \frac{h(f(x))}{f(x)} m_{1/f}(dx).
\end{eqnarray}
\end{proposition}
\begin{proof}
The law of $T_{0}$ under $P_{x}^{1/f}, x>0$ is the same as that of $l_{x}=\sup \{t > 0: X_{t}=x\}$ under $P_{0}^{f}$
according to Exercise 4.18 in Chapter VII of \citet{revuzyor}.
Thus, the right-hand side of the conclusion is equal to
\begin{eqnarray}
\int_{0}^{\infty} \lambda(x)P_{0}^{f} [\rho(l_{x})] f(x)h(f(x)) m_{f}(dx).
\end{eqnarray}

Let $q^{f}(t,x,\cdot)$ be the density of $X_{t}$ under $P_{x}^{f}$ with respect to the speed measure $m_{f}$.
Exercise 4.16 in Chapter VII of \citet{revuzyor} can be applied to $X$ under $P_{j}^{f}, j > 0$.
Subsequently, the law of $l_{x}$ under $P_{j}^{f}$ is expressed as
\begin{eqnarray}
P_{j}^{f}[l_{x}\in dt] = \frac{q^{f}(t,j,x)}{f(x)} dt.
\end{eqnarray}
Let $l_{x}^{T}=\sup \{t < T: X_{t}=x\}$ for any $T > 0$.
Then, we obtain
\begin{eqnarray}
P_{j}^{f}[l_{x}^{T} \neq l_{x}]
= P_{j}^{f}[l_{x} > T]
\le P_{0}^{f}[l_{x} > T] \xrightarrow[T \rightarrow \infty]{} 0
\end{eqnarray}
for $j > 0$.
The convergence $\lim_{j\downarrow 0}P_{j}^{f} [\rho(l_{x})]=P_{0}^{f} [\rho(l_{x})]$
follows from
\begin{eqnarray}
\left|P_{j}^{f} [\rho(l_{x})] - P_{0}^{f} [\rho(l_{x})] \right|
\le 
\left|P_{j}^{f} [\rho(l_{x}^{T})] - P_{0}^{f} [\rho(l_{x}^{T})] \right|
+ 4 ||\rho||_{\infty}P_{0}^{f}[l_{x} > T].
\end{eqnarray}
Therefore, we obtain
\begin{eqnarray}
\int_{0}^{\infty} \lambda(x)P_{j}^{f} [\rho(l_{x})] f(x)h(f(x)) m_{f}(dx)
&=& \int_{0}^{\infty} \lambda(x)\left(\int_{0}^{\infty}\rho(t)\frac{q^{f}(t,j,x)}{f(x)} dt \right) f(x)h(f(x)) m_{f}(dx)\nonumber\\
&=& \int_{0}^{\infty} \rho(t)\left(\int_{0}^{\infty}\lambda(x)h(f(x)) q^{f}(t,j,x) m_{f}(dx)\right) dt  \nonumber\\
&=& \int_{0}^{\infty} \rho(t)P_{j}^{f}[\lambda(X_{t})h(f(X_{t}))] dt.
\end{eqnarray}
The conclusion follows by letting $j \downarrow 0$.
\end{proof}

\subsection{Forward prices as probabilities}
\label{sec:as_prob}
In contrast to the case of a true martingale,
the forward price in the presence of a bubble is a concave function with respect to the underlying price.
This observation, together with the work of \citet{profeta2010}, inspires us to investigate the distributions that are implied in the forward price.
We use the boundary condition (\ref{eq:boundary})
and Doob's maximal identity, which is employed for deriving the expressions of European put and call option prices in terms of last-passage times (Theorems 2.1 and 2.2 of \citet{profeta2010}).

Prior to the derivation, we introduce a generalization of the law of the Bessel meander.
Suppose that $v^{*}(\tau,\infty) < \infty$, and let $M_{\tau}^{1/f}$ be the probability measure on $\mathcal{F}_{\tau}^{0}$ that is defined by
\begin{eqnarray}
M_{\tau}^{1/f} := \frac{f(X_{\tau})}{P_{0}^{f}[f(X_{\tau})]}P_{0}^{f}.
\end{eqnarray}
In the case of $f(x)=x^{-2\nu}, \nu > 0$,
the law of $X$ under the probability $M_{\tau}^{1/f}$ is that of the Bessel meander (see Chapter 3.6 of \citet{mansuy2008}).
The coordinate process $X$ satisfies the stochastic differential equation
\begin{eqnarray}
X_{t} = \tilde{\beta_{t}} - \frac{1}{2} \int_{0}^{t} T_{f}(X_{u})du
+\int_{0}^{t} m(\tau-u,X_{u}) du, \; t \le \tau, 
\end{eqnarray}
where $\tilde{\beta}$ is an $M_{\tau}^{1/f}$-Brownian motion starting at $0$
and $m(t,x) := \frac{d}{dx}\log v^{*}(t,f(x))$.
The integrand of the drift term has the following expression:
\begin{eqnarray}
  -\frac{1}{2}T_{1/f}(X_{u})
  + \left. \frac{d}{dx} \log P_{x}^{1/f}[\tau-u<T_{0}] \right|_{x=X_{u}}.
\end{eqnarray}

The forward prices are interpreted in terms of the probabilities $P_{0}^{f}$, $P_{\cdot}^{1/f}$, and $M_{\tau}^{1/f}$, as in the following proposition.
We remark that the sublinearity (\ref{eq:sublinear}) for each $\tau>0$ has already been obtained in Corollary 6.1 of \citet{Cetin2018}.
\begin{proposition}
\label{prop:distrib}
For $\tau > 0$ and $y = f(x) > 0$, we obtain
\begin{eqnarray}
v_{y}^{*}(\tau,y) &=& P_{0}^{f}[X_{\tau} \le x],\\
\frac{v^{*}(\tau,y)}{y} &=& P_{0}^{f}\left[\inf_{u>\tau}X_{u} \le x \right],\\
\frac{v^{*}(\tau,y)}{v^{*}(\tau,\infty)} &=& M_{\tau}^{1/f}\left[P_{X_{\tau}}^{1/f}\left[\sup_{u>0}X_{u} > x \right]\right],
\end{eqnarray}
where we assume $v^{*}(\tau,\infty)<\infty$ for the third equation.
In particular, $v^{*}(\tau,y)$ is of strictly sublinear growth in $y$ for each $\tau >0$:
\begin{eqnarray}
\lim_{y \rightarrow \infty} \frac{v^{*}(\tau,y)}{y} = 0\label{eq:sublinear}
\end{eqnarray}
and $v_{y}^{*}(\tau,\infty)=0$.
\end{proposition}
\begin{proof}
The boundary condition (\ref{eq:boundary}) with $h=1_{[y,\infty)}$ implies that
\begin{eqnarray}
P_{0}^{f}[f(X_{\tau}) \ge y] = -\int_{y}^{\infty} v_{yy}^{*}(\tau,\eta)d\eta
= v_{y}^{*}(\tau,y)-v_{y}^{*}(\tau,\infty),
\end{eqnarray}
which yields $v_{y}^{*}(\tau,0)=1$, $v_{y}^{*}(\tau,\infty)=0$ and the first equality.
Then, the boundary condition (\ref{eq:boundary}) and integration-by-parts formula (\ref{eq:ibp2}) yield
\begin{eqnarray}
v^{*}(\tau,y) = -\int_{0}^{\infty} \left(y \wedge \eta \right)v_{yy}(\tau,\eta) d\eta
= P_{0}^{f}[y \wedge f(X_{\tau})].
\end{eqnarray}
We note that $f(X)$ is a non-negative local martingale under the probability measure $P_{x}^{f}$ that converges $P_{x}^{f}$ almost surely to $0$ as $t \longrightarrow \infty$.
Based on Doob's maximal identity (see, for example, II.3.12 of \citet{revuzyor}), we obtain
\begin{eqnarray}
P_{\xi}^{f}\left[\sup_{u>0}f(X_{u}) \ge y \right] = 1 \wedge \frac{f(\xi)}{y}
\end{eqnarray}
for $\xi > 0$
and
\begin{eqnarray}
\frac{v^{*}(\tau,y)}{y}
= \frac{1}{y}P_{0}^{f}[f(X_{\tau})\wedge y]
= P_{0}^{f}\left[P_{X_{\tau}}^{f}\left[\inf_{u>0}X_{u} \le x \right] \right]
= P_{0}^{f}\left[\inf_{u>\tau}X_{u} \le x \right].
\end{eqnarray}
This quantity tends to $0$ as $y\longrightarrow \infty$
because of $P_{0}^{f}[X_{\tau}>0, \lim_{u \rightarrow \infty}X_{u} = \infty]=1$.
Finally, the left-hand side of the third formula is expressed as
\begin{eqnarray}
\frac{v^{*}(\tau,y)}{v^{*}(\tau,\infty)}
= P_{0}^{f}\left[\left(1 \wedge \frac{y}{f(X_{\tau})}\right) \frac{f(X_{\tau})}{P_{0}^{f}[f(X_{\tau})]}\right]
= M_{\tau}^{1/f}\left[1 \wedge \frac{1/f(X_{\tau})}{1/y}\right].
\end{eqnarray}
By applying Doob's maximal identity with $f$ replaced with $1/f$, we obtain
\begin{eqnarray}
\frac{v^{*}(\tau,y)}{v^{*}(\tau,\infty)}
= M_{\tau}^{1/f}\left[P_{X_{\tau}}^{1/f}\left[\sup_{u>0} 1/f(X_{u}) > 1/y \right]\right]
= M_{\tau}^{1/f}\left[P_{X_{\tau}}^{1/f}\left[\sup_{u>0} X_{u} > x \right]\right].
\end{eqnarray}
\end{proof}

\subsection{Finiteness of $v^{*}(\tau,\infty)$}
\label{sec:finiteness}
We consider the problem of $v^{*}(\tau,\infty) < \infty$ for $\tau > 0$.
Theorem 4.1 of \citet{Tysk2009} provides a sufficient condition, which states that
if constants $\varepsilon > 0$ and $p > 2$ exist such that
\begin{eqnarray}
\sigma^{2}(y) \ge \varepsilon y^{p}
\end{eqnarray}
for all $y>1/\varepsilon$,
$v^{*}(\tau,y) = o(y^{\delta})$ for any $\tau > 0$ and $\delta > 0$,
and $v^{*}(\tau,\infty) < \infty$ if the above inequality holds for $p>3$.

We obtain the following criterion, which is satisfied by the above case with $\lambda = p-2$.
Refer to (\ref{eq:cev_power}) for an explicit example.
\begin{proposition}
Suppose that a constant $\lambda >0$ exists
such that
\begin{eqnarray}
\int_{1}^{\infty} \eta^{q} \frac{\eta \wedge y}{\sigma(\eta)^{2}} d\eta = O(y^{(q-\lambda)_{+}})
\end{eqnarray}
for all $q \in [0,1+\lambda)$.
Then, we obtain $v^{h}(\tau,\infty) < \infty$ for each $\tau >0$ and $h(y)=y^{\alpha}$ with $\alpha \in [0,1+\lambda)$.
\end{proposition}
\begin{proof}
Let $h$ be a bounded, non-negative continuous function with $h(y)=0$ for $y \in [0,1]$.
Then, for $0 < \delta < \tau$, Theorem \ref{thm:main} yields
\begin{eqnarray}
\frac{1}{2}\int_{\delta}^{\tau} v^{h}(u,\infty) du
&=& -\int_{\delta}^{\tau} \left(\int_{0}^{\infty} v_{\tau}^{*}(u,y)\frac{h(y)}{\sigma(y)^{2}} dy\right) du \nonumber\\
&=& \int_{0}^{\infty} (v^{*}(\delta,y)-v^{*}(\tau,y))\frac{h(y)}{\sigma(y)^{2}} dy \nonumber\\
&\le& \int_{1}^{\infty} v^{*}(\delta,y)\frac{h(y)}{\sigma(y)^{2}} dy\nonumber\\
&=& \int_{1}^{\infty} P_{0}^{f}[f(X_{\delta})\wedge y]\frac{h(y)}{\sigma(y)^{2}} dy.
\end{eqnarray}
The final step is based on Proposition \ref{prop:distrib}.
According to the monotone convergence theorem, 
the above inequality is valid for $h(y)=(y-1)_{+}^{q}$ for $q \in [0,1+\lambda)$.
Furthermore, $a_{q},b_{q}>0$ exist such that
\begin{eqnarray}
\frac{1}{2}\int_{\delta}^{\tau} v^{h}(u,\infty) du
\le P_{0}^{f}\left[\int_{1}^{\infty} h(y) \frac{(f(X_{\delta})\vee 1)\wedge y}{\sigma(y)^{2}} dy\right]
\le a_{q} P_{0}^{f}\left[f(X_{\delta})^{(q-\lambda)_{+}}\right] + b_{q}.
\end{eqnarray}
For a given $\alpha \in [0,1+\lambda)$, let $N$ be the integer such that $\alpha -N\lambda >0 \ge \alpha -(N+1)\lambda$,
and $q_{n}=\alpha-(N-n)\lambda$ for $n=0,1,\dots,N$.
By induction, $v^{h}(\tau,\infty) < \infty$ holds for almost all $\tau > 0$, where $h(y):=y^{\alpha}$.
For any $\tau >0$, $\delta \in (0,\tau)$ exists such that $v^{h}(\delta,\infty) < \infty$.
For any $x > 0$, we obtain
\begin{eqnarray}
v^{h}(\tau,f(x)) = P_{x}^{f}[v(\delta,f(X_{\tau-\delta}))] \le v^{h}(\delta,\infty).
\end{eqnarray}
\end{proof}

Under the assumption $v^{*}(\tau,\infty) < \infty$,
we present two results regarding the convergences of $v^{*}(\tau,y)$ and its average to $v^{*}(\tau,\infty)$.
\begin{proposition}
\label{prop:convergence_rate}
Suppose that $v^{*}(\tau,\infty) < \infty$ holds.
Then, we obtain
\begin{eqnarray}
\left(v^{*}(\tau,\infty)-v^{*}(\tau,y)\right)
+ yv_{y}^{*}(\tau,y) = P_{0}^{f}[f(X_{\tau}):f(X_{\tau}) > y]
\end{eqnarray}
and
\begin{eqnarray}
\frac{v^{*}(\tau,\infty)+v^{*}(\tau,y)}{2}
-\int_{0}^{y}v^{*}(\tau,\eta)\frac{d\eta}{y}
= \frac{1}{2}P_{0}^{f}\left[f(X_{\tau}) \left(1\wedge \frac{f(X_{\tau})}{y}\right)\right]
\end{eqnarray}
for any $\tau > 0$ and $y>0$.
In particular, $\lim_{y\rightarrow \infty} yv_{y}^{*}(\tau,y)=0$.
\end{proposition}
\begin{proof}
From Proposition \ref{prop:doubleintegral}, we obtain
\begin{eqnarray}
P_{0}^{f}[f(X_{\tau}): f(X_{\tau}) \le y]
&=& -\int_{f^{-1}(y)}^{\infty} \frac{d}{d\tau} P_{x}^{1/f}[\tau < T_{0}] f(x)^{2} m_{f}(dx) \nonumber\\
&=& -\int_{f^{-1}(y)}^{\infty} \frac{d}{d\tau} v^{*}(\tau,f(x)) f(x) m_{f}(dx).
\end{eqnarray}
By changing the variable $y=f(x)$ and (\ref{eq:bse}), this is equal to
\begin{eqnarray}
-\int_{0}^{y} v_{\tau}^{*}(\tau,\eta) \eta \frac{2d\eta}{\sigma(\eta)^{2}}
= -\int_{0}^{y} v_{yy}^{*}(\tau,\eta) \eta d\eta
= - yv_{y}^{*}(\tau,y) + v^{*}(\tau,y),
\end{eqnarray}
which yields the first formula.

Using the integration-by-parts and first formulas, we obtain
\begin{eqnarray}
2\int_{0}^{y}v^{*}(\tau,\eta)\frac{d\eta}{y}
&=& \int_{0}^{y} \frac{v^{*}(\tau,\eta)}{\eta}(\eta^{2})^{\prime}\frac{d\eta}{y}\nonumber\\
&=& v^{*}(\tau,y)+v^{*}(\tau,\infty)-\int_{0}^{y} (v^{*}(\tau,\infty)-v^{*}(\tau,\eta)+\eta v_{y}^{*}(\tau,\eta))\frac{d\eta}{y} \nonumber\\
&=& v^{*}(\tau,y)+\int_{0}^{y} P_{0}^{f}[f(X_{\tau}): f(X_{\tau}) > \eta]\frac{d\eta}{y}.
\end{eqnarray}
Then, the left-hand side of the second equation is 
\begin{eqnarray}
\frac{1}{2}\int_{0}^{y} P_{0}^{f}[f(X_{\tau}):X_{\tau} < f^{-1}(\eta)] \frac{d\eta}{y}
= \frac{1}{2}P_{0}^{f}\left[f(X_{\tau})\left(1 \wedge \frac{f(X_{\tau})}{k}\right)\right].
\end{eqnarray}
\end{proof}  

\subsection{Uniqueness}
\label{sec:uniqueness}
\citet{Tysk2009} provide a sufficient condition for the Black--Scholes equation (\ref{eq:bse}) to have a unique solution.
Theorem 4.3 of \citet{Tysk2009} claims that, for an initial condition $h$ of strictly sublinear growth, 
there exists a unique solution among the class of functions of strictly sublinear growth.
In this case, a function $h:[0,\infty) \longrightarrow \mathbb{R}$ is of strictly sublinear growth if $\lim_{y\rightarrow \infty}\frac{h(y)}{y}=0$,
and a function $v:[0,T] \times [0,\infty) \longrightarrow \mathbb{R}$ is of strictly sublinear growth
if $|v(\tau,y)| \le g(y)$ for some $g:[0,\infty) \longrightarrow \mathbb{R}$ of strictly sublinear growth.

\citet{Cetin2018} extends this to an initial condition $h$ of at most linear growth,
expanding the class of solutions to that of functions that satisfy
$\lim_{n\rightarrow \infty}v(\tau_{n},y_{n})/y_{n} = 0$ if $y_{n} \nearrow \infty$ and $\tau_{n}\rightarrow \tau$ for each $\tau > 0$
and
\begin{eqnarray}
\sup_{(\tau,\eta) \in [\tau_{1},\tau_{2}]\times (y,\infty)}\frac{v(\tau,\eta)}{\eta} < \infty\label{eq:unif_growth}
\end{eqnarray}
for any $y>0$ and $0 \le \tau_{1} < \tau_{2} < \infty$\footnote{\citet{Cetin2018} states that the second condition, which is cited as (D.3) in \citet{Cetin2018}, follows from the first one.
Refer to Theorem 6.2 and Appendix D of \citet{Cetin2018} for the precise statement and proof.}.
We present two non-negative classical solutions to (\ref{eq:bse}) with $\sigma(y)=y^{2}$ and the initial condition $h(y)=0$
such that at least one of the above conditions are violated:
\begin{eqnarray}
\Delta^{(0)}(\tau,y) &=& \tau^{-\frac{3}{2}} \mathrm{e}^{-\frac{1}{2y^{2}\tau}}, \\
\Delta^{(1)}(\tau,y) &=& \frac{y}{\sqrt{\tau}} \mathrm{e}^{-\frac{1}{2y^{2}\tau}}.
\end{eqnarray}
The asymptotic behaviors as $y \longrightarrow \infty$ for each $\tau>0$ are $\Delta^{(0)}(\tau,\infty) = \tau^{-\frac{3}{2}}$
and $\Delta^{(y)}(\tau,y) \sim y /\sqrt{\tau}$,
whereas $\Delta^{(0)}(\tau,y) \sim y^{3}$ and $\Delta^{(1)}(\tau,y) \sim y^{2}$ as $(\tau,y) \longrightarrow (0,\infty)$
such that $y\sqrt{\tau}$ is constant.

Our boundary condition (\ref{eq:boundary}) causes the Black--Scholes equation (\ref{eq:bse}) to have a unique solution among an appropriate class of considered functions.
We remark that 
our boundary condition (\ref{eq:boundary}) with the case $h(y)=y$; that is,
\begin{eqnarray}
v(\tau,\infty)
= -\int_{0}^{\infty} yv_{yy}(\tau,y)dy,  \label{eq:f_boundary}
\end{eqnarray}
holds with inequality
for a function with one variable under certain assumptions, as follows:
\begin{proposition}
Let $g:[0,\infty) \longrightarrow [0,\infty)$ be in $C^{2}((0,\infty))$
with $g_{yy} \le 0$
and $\lim_{y\rightarrow 0} yg_{y}(y)=0$.
Then, we obtain
\begin{eqnarray}
g(\infty) \ge g(\infty)-g(0) \ge - \int_{0}^{\infty} yg_{yy}(\tau,y) dy. \label{eq:boundary_ineq}
\end{eqnarray}
In the case of $g(\infty)<\infty$,
equality holds in the second inequality if and only if $\lim_{y\rightarrow \infty}yg_{y}(y)=0$.
\end{proposition}
\begin{proof}
We remark that $g_{y}(\tau,y) \ge 0$, which follows from $g \ge 0$ and $g_{yy} \le 0$.
Using the integration-by-parts formula, we obtain
\begin{eqnarray}
g(\tau,\infty)-g(\tau,0)
\ge \lim_{y\rightarrow \infty} (g(y)-yg_{y}(y))-g(0)
= - \int_{0}^{\infty} yg_{yy}(y) dy.
\end{eqnarray}
\end{proof}
Then, the duality relation
\begin{eqnarray}
\inf_{v} v(\tau,\infty)
\ge \sup_{v} \left\{ -\int_{0}^{\infty} yv_{yy}(\tau,y)dy \right\}
\end{eqnarray}
holds among an appropriate class of functions $v$,
for which the optimal values are attained by the stochastic solution $v^{h}$. 
The equality, which corresponds to our boundary condition (\ref{eq:f_boundary}),
holds for $v^{h}$ and $\Delta^{(0)}$,
and does not hold for $v(\tau,y)=y$ and $v=\Delta^{(1)}$:
\begin{eqnarray}
\infty = v(\tau,\infty) > - \int_{0}^{\infty} yv_{yy}(\tau,y) dy = 0.
\end{eqnarray}
To achieve uniqueness by adding our boundary condition (\ref{eq:boundary}),
a specific condition is required to exclude functions such as $\Delta^{(0)}$.
The following proposition is one criterion for uniqueness, in which $\Delta^{(0)}$ does not satisfy (\ref{eq:fubini}).

\begin{proposition}
Suppose that $v$ is a non-negative classical solution to (\ref{eq:bse}) with $v(0,y)=h(y)$
and that 
$\Delta := v - v^{h}$ satisfies
\begin{eqnarray}
\int_{0}^{\tau} \int_{0}^{\infty} |\Delta_{\tau}(u,y)|\frac{y}{\sigma(y)^{2}} dy du < \infty \label{eq:fubini}
\end{eqnarray}
for each $\tau>0$.
Then, we obtain $v=v^{h}$, if $-\int_{0}^{\infty} \Delta_{\tau}(\tau,y)\frac{y}{\sigma(y)^{2}} dy \ge 0$ for each $\tau>0$,
which is the case if $v^{*}(\tau,\infty) < \infty$ and $v$ satisfies (\ref{eq:f_boundary}).
\end{proposition}
\begin{proof}
We remark that $\Delta$ is a non-negative solution to (\ref{eq:bse}) with $\Delta(0,y)=0$.
According to Fubini's theorem, we obtain
\begin{eqnarray}
0 \le -2 \int_{0}^{\tau} \left(\int_{0}^{\infty} \Delta_{\tau}(u,y)\frac{y}{\sigma(y)^{2}} dy \right)du
= -2 \int_{0}^{\infty} \Delta(\tau,y)\frac{y}{\sigma(y)^{2}} dy 
\le 0
\end{eqnarray}
for each $\tau>0$.
\end{proof}

We also provide a criterion for uniqueness in terms of integrability, rather than our boundary condition (\ref{eq:boundary}).
We remark that 
the second integral of the left-hand side of (\ref{eq:uicondition}) is finite
if $h$ and $v$ are of at most liner growth
according to the presence of a bubble; that is, (c) of Assumption \ref{ass:ass_vol}.
\begin{proposition}
\label{prop:uniqueness}
Suppose that $v^{h}$ is in $C^{1,2}((0,\infty)\times (0,\infty))$,
that $v$ is a non-negative classical solution to (\ref{eq:bse}) with $v(0,y)=h(y)$,
and that 
$\Delta := v - v^{h}$ satisfies
\begin{eqnarray}
\int_{0}^{\tau} \int_{0}^{\infty} |\Delta_{\tau}(u,y)|\frac{dy}{\sigma(y)^{2}} du < \infty \label{eq:fubini2}
\end{eqnarray}
for each $\tau>0$.
Then, we obtain $v=v^{h}$ if
\begin{eqnarray}
\Delta(\tau,0)=\lim_{y\rightarrow 0}yv_{y}(\tau,y)=v_{y}(\tau,\infty)=0
\end{eqnarray}
and
\begin{eqnarray}
\int_{0}^{\tau} \left(\sup_{\eta \in (0,y)} \Delta(u,\eta)\right) du
+\int_{y}^{\infty} \left(\sup_{u \in (0,\tau]} \Delta(u,\eta)\right) \frac{d\eta}{\sigma(\eta)^{2}}
< \infty \label{eq:uicondition}
\end{eqnarray}
for each $\tau>0$ and $y > 0$.
\end{proposition}
\begin{proof}
We remark that $\Delta$ is a non-negative classical solution to (\ref{eq:bse}) with $\Delta(0,y)=0$.
Using the integration-by-parts formula (\ref{eq:ibp2}), we obtain
\begin{eqnarray}
\Delta(\tau,y)
= -\int_{0}^{\infty} (y\wedge \eta) \Delta_{yy}(\tau,\eta)d\eta
\end{eqnarray}
for $\tau>0$ and $y>0$.
For $\delta \in (0,\tau)$ and $\varepsilon \in (0,y)$, we obtain
\begin{eqnarray}
\int_{\delta}^{\tau} \Delta(u,y)du
- \int_{\delta}^{\tau} \Delta(u,\varepsilon)du
&=& 
2 \int_{0}^{\infty} \left(\Delta(\delta,\eta)-\Delta(\tau,\eta)\right) \frac{y\wedge \eta-\varepsilon \wedge \eta}{\sigma(\eta)^{2}}d\eta \nonumber\\
&\le&
2 (y-\varepsilon)\int_{\varepsilon}^{\infty} \Delta(\delta,\eta) \frac{d\eta}{\sigma(\eta)^{2}}.
\end{eqnarray}
By letting $\delta \downarrow 0$, and then letting $\varepsilon \downarrow 0$,
the dominated convergence theorem yields
\begin{eqnarray}
0 \le \int_{0}^{\tau} \Delta(u,y)du
\le \int_{0}^{\tau} \Delta(u,\varepsilon)du \longrightarrow 0.
\end{eqnarray}
\end{proof}

\section{Example}
\label{sec:example}
\subsection{Constant elasticity of variance model}
\label{sec:cev}
The constant elasticity of variance (CEV) model is expressed as
\begin{eqnarray}
Y_{t} = y + \int_{0}^{t} a Y_{u}^{1+\frac{1}{2\nu}} d\beta_{u}, \label{eq:cev}
\end{eqnarray}
where $a>0$ and $\nu \neq 0$ are constants; that is, $\sigma(y) = a y^{1+\frac{1}{2\nu}}$.
The local martingale $Y$ is a strict local martingale if and only if $\nu >0$, which we assume hereafter.
\citet{Veestraeten2017} studies option prices and their Greeks extensively, emphasizing the multiplicity of call option prices under this model.

From $\sigma(y) = a y^{1+\frac{1}{2\nu}}$, we obtain 
$f(x)=(2\nu/a)^{2\nu}x^{-2\nu}$,
$f^{-1}(y)=(2\nu/a)y^{-1/(2\nu)}$,
$T_{f}(x)=-(1+2\nu)/x$, and
$T_{1/f}(x)=-(1-2\nu)/x$.
Then, the law of the underlying process $Y$ is that of $f(X)$,
where $X$ is a Bessel process with index $\nu$ under $P_{x}^{f}$.
The stochastic solution to the corresponding Black--Scholes equation for $h(y)=y$ is
\begin{eqnarray}
v^{*}(\tau,y) = y \int_{0}^{s} \eta^{\nu-1} \mathrm{e}^{-\eta/2} \frac{d\eta}{2^{\nu} \Gamma(\nu)}
= \frac{1}{\nu \Gamma(\nu)}\left(\frac{2\nu^{2}}{a^{2}\tau}\right)^{\nu}
\int_{0}^{1} \mathrm{e}^{-\frac{2\nu^{2}}{a^{2}\tau}\left(\frac{\zeta}{y}\right)^{1/\nu}}d\zeta\label{eq:cev_forward}
\end{eqnarray}
and its derivative with respect to $y$ is
\begin{eqnarray}
v_{y}^{*}(\tau,y)
= \int_{0}^{s} \eta^{\nu-1} \mathrm{e}^{-\eta/2} \frac{d\eta}{2^{\nu} \Gamma(\nu)}
-\frac{s^{\nu} \mathrm{e}^{-s/2}}{2^{\nu} \nu \Gamma(\nu)},
\end{eqnarray}
where $s := (2\nu/a)^{2}/(y^{1/\nu}\tau)$ and $\Gamma$ is the gamma function.
By letting $y\rightarrow \infty$ in the second expression of (\ref{eq:cev_forward}),
we obtain
\begin{eqnarray}
v^{*}(\tau,\infty)
=  \frac{1}{\nu\Gamma(\nu)} \left(\frac{2\nu^{2}}{a^{2}\tau}\right)^{\nu}.\label{eq:cev_forward_inf}
\end{eqnarray}
The probabilities in Proposition \ref{prop:distrib} are obtained explicitly.
In particular,
\begin{eqnarray}
\frac{v^{*}(\tau,y)}{y} = P_{0}^{f}\left[\inf_{u>\tau}X_{u} \le x \right]
= \int_{0}^{s} \eta^{\nu-1} \mathrm{e}^{-\eta/2} \frac{d\eta}{2^{\nu} \Gamma(\nu)},
\end{eqnarray}
which is also equal to $P_{x}^{1/f}[\tau < T_{0}]$; that is,
the complementary distribution function for a Bessel process with index $-\nu$ starting from $x$ to reach $0$.

The asymptotic behaviors of $v^{*}(\tau,y)/y$ and $v_{y}^{*}(\tau,y)$ are determined by how $(\tau,y)$ approaches $(0,\infty)$,
because they are dependent only on $s$.
For a fixed $y > 0$, $v^{*}(\cdot,y)$ decreases from $v^{*}(0,y)=y$ to $v^{*}(\infty,y)=0$,
and for a fixed $\tau > 0$, $v^{*}(\tau,\cdot)$ increases from $v^{*}(\tau,0)=0$ to $v^{*}(\tau,\infty)$.
The boundary condition proposed by \citet{EkstromLotstedtSydowTysk2011} is reasonable
because $v_{y}^{*}(\tau,y)$ tends to $0$ as $y \rightarrow \infty$ for each $\tau>0$.
However, the convergence is not uniform
and $v_{y}^{*}(\tau,y)$ tends to $1$ as $\tau \downarrow 0$ for each $y > 0$.

Suppose that $a=1$.
The derivative price that pays $h(y) = y^{p}$, $p \in [0,1+1/\nu)$ with time-to-maturity $\tau>0$ is 
\begin{eqnarray}
v^{h}(\tau,y)
= \frac{\Gamma(\nu+1-\nu p)}{\Gamma(\nu+1)}
\left(\frac{x}{2\tau}\right)^{\nu p} \Phi \left(\nu p,\nu+1,-\frac{x}{2\tau} \right)y^{p}, \label{eq:cev_power}
\end{eqnarray}
where $\Phi$ is Kunner's function.
Then, a constant $C_{\nu,p} > 0$ exists such that $v^{h}(\tau,y) < C y^{p}$ for all $\tau > 0$ and $y > 0$, owing to $\sup_{x>0} x^{a} \Phi (a,b,-x) < \infty$.
This indicates that the second integral of (\ref{eq:uicondition}) is finite for a function $v$ such that $\sup_{u \in [0,\tau), y > 1} v(\tau,y)/y^{p} < \infty$.
Uniqueness holds among the class of functions that satisfy (\ref{eq:fubini2}) and the assumptions pertaining to the neighborhood of $0$ of Proposition \ref{prop:uniqueness}.
In the case of $\nu=1/2$,
non-negative classical solutions $v^{h}+ \alpha \Delta^{(i)}$ for $i=0,1$ and $\alpha > 0$ are excluded by the boundary condition (\ref{eq:boundary})
if $h$ is not identical to $y$.
The non-negative solution $v = v^{*}+ \alpha \Delta^{(0)}$ for $h(y)=y$ satisfies the boundary condition (\ref{eq:f_boundary}).
Our numerical procedure presented in Section \ref{sec:ourprocedure} excludes this solution because our solution is non-increasing with respect to $\tau$ (see Proposition \ref{prop:discreteanalysis}), 
whereas $v_{\tau}$ takes positive values:
\begin{eqnarray}
v_{\tau}(\tau,y)
= \left(-\sqrt{\frac{2}{\pi}}+\frac{\alpha}{2\tau}\left(\frac{1}{y^{2}\tau}-3\right) \right)\Delta^{(0)}(\tau,y).
\end{eqnarray}

\subsection{Quadratic normal volatility model}
The quadratic normal volatility model is expressed as
\begin{eqnarray}
Y_{t} = y + \int_{0}^{t} b \frac{(Y_{u}-r)(Y_{u}-l)}{r-l} d\beta_{u}
\end{eqnarray}
for $l < r$ and $b>0$.
This model has attracted significant attention because it allows explicit formulas for option prices and can be effectively calibrated to market volatility smiles.
\citet{Andersen2011} derives pricing formulas for all root configurations of the diffusion coefficient,
and \citet{CarrFisherRuf2013} investigate its properties.
Hereafter, we assume that $y > r=0$, in which case $Y$ is a strict local martingale and strictly positive.

From $\sigma(y) = -\frac{b}{l} y(y-l)$, we obtain 
$f(x)= -\frac{l}{2}(1-\coth \frac{b}{2}x)$,
$x=f^{-1}(y)=\frac{1}{b}\log\frac{y-l}{y}$,
$T_{f}(x)=-b \coth \frac{b}{2}x$, and
$T_{1/f}(x)=b$.
The diffusion $X$ under $P_{x}^{1/f}$ is a Brownian motion with drift $-b/2$.
The stochastic solution to the corresponding Black--Scholes equation for $h(y)=y$ is
\begin{eqnarray}
v^{*}(\tau,y) =y - (y-l)N(d_{-})-yN(d_{+})
= y \left(N(-d_{-})-N(d_{+})\right) + l N(d_{-}) \label{eq:qvm_forward}
\end{eqnarray}
and its derivative with respect to $y$ is
\begin{eqnarray}
v_{y}^{*}(\tau,y)
= \left(N(-d_{-})-N(d_{+})\right)
+ \frac{-l/\sqrt{b^{2}\tau}}{y(y-l)}\left(y \left(N^{\prime}(-d_{-})-N^{\prime}(d_{+})\right) + l N^{\prime}(d_{-}) \right),
\end{eqnarray}
where $N$ is the distribution function of the standard normal distribution and
\begin{eqnarray}
d_{\pm} = \frac{\log \frac{y}{y-l} \pm \frac{b^{2}\tau}{2}}{\sqrt{b^{2}\tau}}= -\frac{x}{\sqrt{\tau}} \pm \frac{b}{2}\sqrt{\tau}.
\end{eqnarray}
By letting $y\rightarrow \infty$ in the second expression of (\ref{eq:qvm_forward}),
we obtain
\begin{eqnarray}
v^{*}(\tau,\infty)
=-l\sqrt{\frac{2}{\pi}} \frac{\mathrm{e}^{-b^{2}\tau/8} }{\sqrt{b^{2}\tau}}
+ lN \left(-\frac{\sqrt{b^{2}\tau}}{2} \right).
\end{eqnarray}
The probabilities in Proposition \ref{prop:distrib} are obtained explicitly.
In particular,
\begin{eqnarray}
\frac{v^{*}(\tau,y)}{y} = P_{0}^{f}\left[\inf_{u>\tau}X_{u} < x \right]
= \left(N(-d_{-})-N(d_{+})\right)+\frac{l}{y}N(d_{-}),
\end{eqnarray}
which is also equal to $P_{x}^{1/f}[\tau < T_{0}]$; that is,
the complementary distribution function for a Brownian motion with drift $-b/2$ starting from $x$ to reach $0$.

The asymptotic behaviors of $v^{*}(\tau,y)/y$ and $v_{y}^{*}(\tau,y)$ are determined by how $(\tau,y)$ approaches $(0,\infty)$.
For instance, let $(\tau,y)$ approach $(0,\infty)$ such that $s=x/\sqrt{\tau}$ is constant.
Then, $v^{*}(\tau,y)/y$ and $v_{y}^{*}(\tau,y)$ converge to
$N(s)-N(-s)$.

\section{Numerical analysis}
\label{sec:numerical}
\subsection{Description of procedures}
\label{sec:procedures}
We describe our numerical procedures
as well as those proposed in the literature to date,
namely in \citet{EkstromLotstedtSydowTysk2011}, \citet{SongYang2015}, \citet{Cetin2018}, and \citet{tsuzuki2023pitmans}.

\subsubsection{Our procedure}
\label{sec:ourprocedure}
We propose boundary conditions for finite difference methods in the case of $h(y)=y$ and in the other cases, respectively.
For the first case, using (\ref{eq:bse}) and (\ref{eq:boundary}),
we obtain
\begin{eqnarray}
v^{*}(\tau,\infty)
&=& -2\int_{0}^{\infty} v_{\tau}^{*}(\tau,y) \frac{y}{\sigma(y)^{2}} dy \nonumber\\
&=& -\int_{0}^{n} v_{yy}^{*}(\tau,y) ydy
-2\int_{n}^{\infty} v_{\tau}^{*}(\tau,y) \frac{y}{\sigma(y)^{2}} dy \nonumber\\
&=& v^{*}(\tau,n) - nv_{y}^{*}(\tau,n)
-2\int_{n}^{\infty} v_{\tau}^{*}(\tau,y) \frac{y}{\sigma(y)^{2}} dy.
\end{eqnarray}
We propose the following as our boundary condition:
\begin{eqnarray}
v_{y}^{*}(\tau,n) = -\varepsilon_{n} v_{\tau}^{*}(\tau,n),\;
\varepsilon_{n} = \frac{2}{n} \int_{n}^{\infty}  \frac{y}{\sigma(y)^{2}} dy,
\label{eq:bc}
\end{eqnarray}
for which the finite difference $\theta$-scheme, $\theta \in [0,1]$, is
\begin{eqnarray}
\theta \Delta_{y} v^{*}(\tau,n) + (1-\theta)\Delta_{y} v^{*}(\tau-\Delta \tau,n)
=
-\varepsilon_{n} \Delta_{\tau} v^{*}(\tau,n),\label{eq:dbc}
\end{eqnarray}
where $\Delta_{\tau} v^{*}(\tau,y)$ and $\Delta_{y} v^{*}(\tau,y)$
are first-order finite differences
with the time- and space-discretization steps $\Delta \tau$ and $\Delta y > 0$, respectively:
\begin{eqnarray}
\Delta_{\tau} v^{*}(\tau,y) &=& \frac{v^{*}(\tau,y)-v^{*}(\tau-\Delta \tau,y)}{\Delta \tau},\\
\Delta_{y} v^{*}(\tau,y) &=& \frac{v^{*}(\tau,y)-v^{*}(\tau,y-\Delta y)}{\Delta y}.
\end{eqnarray}
We remark that $n\varepsilon_{n} \longrightarrow 0$ owing to (c) of Assumption \ref{ass:ass_vol}.
Condition (\ref{eq:bc}) corresponds to a solution such that
the left-derivative of $v^{*}(\tau,\cdot)$ at $y=n$ 
is $\lim_{y \nearrow n}v_{y}^{*}(\tau,y)$
and
$v_{y}^{*}(\tau,y) = 0$ for $y>n$.

For the other cases,
we propose the following value as our boundary condition for $v^{h}(\tau,\cdot)$ at $y=n$:
\begin{eqnarray}
v^{h}(\tau,\infty)
= -2\int_{0}^{n} \Delta_{\tau} v^{*}(\tau,y) \frac{h(y)}{\sigma(y)^{2}} dy
-n \varepsilon_{n} \Delta_{\tau} v^{*}(\tau,n), \label{eq:eq:dbc2}
\end{eqnarray}
where $\Delta_{\tau}v^{*}(\tau,y)$ is computed in the above manner.
As remarked in Section \ref{sec:main},
we recommend this condition rather than that based on (\ref{eq:boundary2})
when the growth rate of $h(y)$ is more than linear.
More precisely, we do not recommend using the boundary condition 
\begin{eqnarray}
v_{y}^{h}(\tau,n) = -\varepsilon_{n}v_{\tau}^{h}(\tau,n),
\end{eqnarray}
which corresponds to (\ref{eq:bc}),
because the truncation error of $v^{h}$ is expected to be more significant 
than that of $v^{*}$.

We must specify the boundary condition at $y=0$ for numerical procedures, even if $Y$ does not reach the boundary.
The appropriate condition is $v(\tau,0) = h(0)$ because the boundary $y=0$ must be an absorbing point for a non-negative supermartingale.
This boundary condition is applied to all methods in this study.
Although strict positivity of $Y$ is assumed in Theorem \ref{thm:main},
our numerical procedure can be performed for $\sigma$ that violates (b) of Assumption \ref{ass:ass_vol},
and the effect for $Y$ to be absorbed at $0$ is considered to a certain extent.

In summary, our numerical procedure for $h(y)=y$ is the finite difference method
for the following partial differential equation on $(0,\infty) \times (0,n)$:
\begin{eqnarray}
\left\{ \begin{array}{ll}
v_{\tau} = \frac{1}{2}\sigma^{2}v_{yy}, \\
v(0,y) = h(y),\\
v(\tau,0) = h(0),\\
v_{y}(\tau,n) = -\varepsilon_{n}v_{\tau}(\tau,n),
\end{array} \right. \label{eq:bce_ourboundary}
\end{eqnarray}
and our numerical procedure for $h$ that is not identical to $y$
is that with the final equation replaced with (\ref{eq:eq:dbc2}).

\subsubsection{\citet{EkstromLotstedtSydowTysk2011}}
\citet{EkstromLotstedtSydowTysk2011} propose the Neumann boundary condition $v_{y}(\tau,n)=0$ for each $\tau>0$.

Although $v_{y}(\tau,\infty)=0$ holds for each $\tau>0$, as shown in Proposition \ref{prop:distrib},
the behavior in the neighborhood of $(0,n)$ is different from the expected behavior; that is,
$v_{y}^{*}(\tau,n)$ is continuously decreasing from $1$ to $0$ as $\tau$ increases from $0$.
Refer to Section \ref{sec:example} for explicit examples.
As a result, the numerical errors in the neighborhood of $(0,n)$, which are propagated throughout the domain, are not negligible.
The underestimation of $v_{y}(\tau,n)$ appears to cause that of the numerical solution.
Their numerical solutions are less than the true values, as demonstrated in Section \ref{sec:numericaltest},
and less than ours, as shown in Proposition \ref{prop:neumann}.
We also remark that
only the right-derivative of $v^{*}(\tau,\cdot)$ at $y=n$ is forced to be $0$ in our procedure.
The implicit finite difference scheme with their boundary condition is stable in the sense that it satisfies a discrete maximum principle.
Refer to Proposition \ref{prop:neumann} in Section \ref{sec:ourstability}.

\subsubsection{\citet{SongYang2015}}
\citet{SongYang2015} propose the condition $v(t,n)= 0$,
revising the initial conditions such that the initial-boundary datum is continuous at the corner $(\tau,y)=(0,n)$.
Although the discontinuity at the corner may cause oscillations,
we do not revise the payoff in this study because the results are sufficiently stable if the implicit scheme is adopted.
Then, the solutions are the prices of knock-out options;
that is, the contingent claim that pays $h(Y_{T})$ at time $T$ only if $\sup_{t < T} Y_{t} < n$.
The price of the knock-out option is the solution to (\ref{eq:bse}) with
\begin{eqnarray}
\left\{ \begin{array}{l}
v(0,y) = h(y),\\
v(\tau,n)=0.
\end{array} \right.
\end{eqnarray}
In particular, we denote the price of the knock-out forward contract as $v^{o,n}$.

The boundary condition $v(\tau,n)= 0$ produces numerical solutions that are substantially less than the true values, which is confirmed by the numerical tests in Section \ref{sec:numericaltest}.
Furthermore, the numerical solution as a function of $y$
is not increasing, but converges to $0$ as $y \longrightarrow n$, which is not a desirable property.
Thus, $n$ should be large to obtain accurate numerical results.

\subsubsection{\citet{Cetin2018}}
As mentioned in Section \ref{sec:problem},
\citet{Cetin2018} reduces the problem to the alternative Cauchy problem (\ref{eq:cetin}) that has a unique solution,
but the straightforward finite scheme of the alternative problem is the same as that of \citet{SongYang2015}.
Instead, we propose altering the problem by changing the variable $y=f(x)$, where $f$ is defined by (\ref{eq:def_f}).
Then, the stochastic solution to the original problem is expressed as $v^{h}(\tau,y)=yu(\tau,f^{-1}(y))$, where $u$ is the solution to 
\begin{eqnarray}
\label{eq:cetin2}
\left\{ \begin{array}{l}
u_{\tau} = \frac{1}{2}u_{xx}-\frac{1}{2}T_{1/f} u_{x}, \\
u(0,x) = h(f(x))/f(x),\\
u(\tau,0)=0.
\end{array} \right.
\end{eqnarray}
In particular, $u^{*}(\tau,x) := v^{*}(\tau,f(x))/f(x)$
is the complementary distribution function for $X$ to reach $0$ with respect to $P_{x}^{1/f}$:
\begin{eqnarray}
u^{*}(\tau,x) = P_{x}^{1/f}\left[\tau < T_{0}\right]. \label{eq:dual}
\end{eqnarray}
Refer to (\ref{eq:sublinear}) for the alternative probabilistic expression of $u^{*}$.

The change in boundary points prevents the accuracy deterioration that is caused by replacing the point at which the boundary condition is set (i.e., infinity) with a sufficiently large but finite value.
However, this may introduce a singularity at $x=0$ in the coefficient of (\ref{eq:cetin2}); that is, $T_{1/f}(0)=\infty$.
Although the existence and uniqueness of the solution are ensured,
whether the standard finite difference schemes for this equation are well behaved is unclear.
More precisely, finite schemes tend to be unstable if (\ref{eq:cetin2}) is convection dominated;
that is, if the convection or drift term is large compared with the diffusion term.
This is the case when Asian options are computed using finite methods (see, for example, \citet{Zvan1997}).
We investigate the stability of the finite schemes in Section \ref{sec:cetinstability}.

The change of variable by $y=f(x)$ is an ad hoc workaround.
Whether the choice of $f$ is optimal is not clear,
and $f$ is not necessarily obtained from $\sigma$ in an explicit form.
If $Y$ reaches the boundary $0$ with a positive probability,
another truncation error arises from imposing a boundary condition at a sufficiently large $x$ on (\ref{eq:cetin2}).

\subsubsection{\citet{tsuzuki2023pitmans}}
\label{sec:tsuzuki2023}
\citet{tsuzuki2023pitmans} provides a convergent sequence $v_{n}^{h} \longrightarrow v^{h}$ as $n \longrightarrow \infty$
for a non-negative, continuous, and non-decreasing $h$,
where $v_{n}^{h}$ is the solution to a Black--Scholes equation on $(0,\infty) \times (0,n)$
with a certain boundary condition at $y=n$ that is expressed with the prices $v^{o,n}$ of knock-out forward contracts.
More precisely, let $v_{n}^{h}$ be the solution to (\ref{eq:bse}) with
\begin{eqnarray}
\left\{ \begin{array}{l}
v(0,y) = h(y),\\
v(\tau,n)=\Theta_{n}^{h}(\tau),
\end{array} \right.
\end{eqnarray}
where
\begin{eqnarray}
\Theta_{n}^{h}(\tau) = -\int_{0}^{n} \lambda_{j}^{\prime\prime}(y) v^{o,n}(\tau,y) dy\label{eq:koexpression}
\end{eqnarray}
and $\lambda_{n}(y):=\left(1-\frac{y}{n}\right)h(y)$, which is assumed in $C^{2}([0,n])$
with $\lambda_{n}(y)v_{y}^{o,n}(\tau,y)=\lambda_{n}^{\prime}(y)v^{o,n}(\tau,y)=0$ for $y=0,n$ and $\tau>0$.
Then, $v_{n}^{h} \nearrow v^{h}$ and $\Theta_{n}^{h}(\tau) \nearrow v^{h}(\tau,\infty)$ as $n \longrightarrow \infty$.

This scheme has two advantages over that of \citet{SongYang2015}.
The first is that the initial-boundary datum is continuous at the corner $(\tau,y)=(0,n)$.
The other is $v^{o,n} \le v_{n}^{h} \le v^{*}$ for $h(y)=y$,
which follows by comparing the boundary values at $y=n$.
However, the numerical tests in Section \ref{sec:numericaltest}
reveal that this method is less accurate than that of \citet{EkstromLotstedtSydowTysk2011}.

\subsection{Stability analysis}
\subsubsection{Stability analysis of \citet{KaratzasRuf2013} and \citet{Cetin2018}}
\label{sec:cetinstability}
We consider a partial differential equation with general coefficients $s$ and $b$:
\begin{eqnarray}
u_{\tau} = \frac{1}{2}s^{2}u_{xx} + bs u_{x}, \; \tau > 0, x > 0.
\end{eqnarray}
Then, the finite difference $\theta$-scheme $\theta \in [0,1]$ around the point $(\tau,x)$ is expressed as
\begin{eqnarray}
\frac{u(\tau,x)-u(\tau^{\prime},x)}{\Delta \tau}
= 
\theta Lu(\tau,x) + (1-\theta)Lu(\tau^{\prime},x),
\end{eqnarray}
where $\tau^{\prime}=\tau-\Delta \tau$ and
\begin{eqnarray}
Lu(\tau,x)
&=&
\frac{s(x)^{2}}{2}
\frac{u(\tau,x+\Delta x)-2u(\tau,x)+u(\tau,x-\Delta x)}{(\Delta x)^{2}} \nonumber\\
&+& b(x)s(x)\frac{u(\tau,x+\Delta x)-u(\tau,x-\Delta x)}{2\Delta x}.
\end{eqnarray}
Arranging the terms leads to
\begin{eqnarray}
u(\tau,x)
&=& 
p u(\tau^{\prime},x) \nonumber\\
&+& q_{+}
(\theta u(\tau,x+ \Delta x)+ (1-\theta) u(\tau^{\prime},x + \Delta x)) \nonumber\\
&+& q_{-}
(\theta u(\tau,x- \Delta x)+ (1-\theta) u(\tau^{\prime},x - \Delta x)),
\end{eqnarray}
where
\begin{eqnarray}
p &=& \frac{1}{\frac{1}{\Delta \tau}+\theta \left(\frac{s(x)}{\Delta x}\right)^{2}}
\left(\frac{1}{\Delta \tau}-(1-\theta) \left(\frac{s(x)}{\Delta x}\right)^{2}\right), \\
q_{\pm} &=& \frac{1}{2}\frac{1}{\frac{1}{\Delta \tau}+\theta \left(\frac{s(x)}{\Delta x}\right)^{2}}
\left(\left(\frac{s(x)}{\Delta x}\right)^{2} \pm \frac{b(x)s(x)}{\Delta x}\right).
\end{eqnarray}
If $p,q_{\pm} \ge 0$, 
which is equivalent to 
$\frac{1}{\Delta \tau} \ge (1-\theta) \left(\frac{s(x)}{\Delta x}\right)^{2}$
and
\begin{eqnarray}
\left(\frac{s(x)}{\Delta x}\right)^{2} \ge \frac{|b(x)s(x)|}{\Delta x},\label{eq:peclet}
\end{eqnarray}
$u(\tau,x)$ is the average of $u(\tau^{\prime},x),u(\tau^{\prime},x\pm \Delta x),u(\tau,x\pm \Delta x)$ with appropriate weights,
and does not admit local maxima/minima.
The condition (\ref{eq:peclet}) is known as the Peclet condition.
If the convection term $b$ is large compared with the diffusion $s$,
the problem is known as convection dominated, where the finite grid must be fine.
See \citet{Zvan1997} and the remark following Theorem 10.1 of \citet{Duffy2006}.

When the method of \citet{Cetin2018} proposed in this study is applied to CEV models; that is, $s(x)=1,b(x)=-(\nu-1/2)/x$, the Peclet condition (\ref{eq:peclet}) is
\begin{eqnarray}
\frac{1}{\Delta x} \ge \frac{\left|2\nu -1\right|}{2x}
=  \frac{\left|2\nu -1\right|}{2j \Delta x},\label{eq:peclet_cev}
\end{eqnarray}
where $x = j \Delta x, j = 1,2,\dots$,
which is violated for $j < |\nu-1/2|$ if $\nu > 3/2$.

We consider the Peclet condition for the case $s(x)=1,b(x)=-cx^{-(1+p)}$ with $c,p>0$.
The solution $u$ with $u(0,x) = 1,u(\tau,0)=0$ is the complementary distribution function of the time-to-explosion of $X$:
\begin{eqnarray}
X_{t} = x + \beta_{t} - c\int_{0}^{t} \frac{du}{X_{u}^{1+p}}.
\end{eqnarray}
The Peclet condition $1/\Delta x \ge c/(j\Delta x)^{1+p}$
is violated for $j < (c/\Delta x^{p})^{1/(1+p)}$, which tends to $\infty$ as $\Delta x \downarrow 0$.

\subsubsection{Stability analysis of our scheme and that of \citet{EkstromLotstedtSydowTysk2011}}
\label{sec:ourstability}
We investigate the stability of our numerical scheme in the case of $h(y)=y$; that is, (\ref{eq:bce_ourboundary}),
because it is the standard finite difference method in the other cases.
We also consider the scheme of \citet{EkstromLotstedtSydowTysk2011}, to which similar arguments can be applied.

The Peclet condition (\ref{eq:peclet}) is satisfied owing to $b=0$.
The $\theta$-scheme of (\ref{eq:dbc}) is
\begin{eqnarray}
\left(\frac{\varepsilon_{n}}{\Delta \tau}+\frac{\theta}{\Delta y}\right) v(\tau,n)
= 
\left(\frac{\varepsilon_{n}}{\Delta \tau}-\frac{1-\theta}{\Delta y}\right) v(\tau^{\prime},n)
+
\frac{\theta}{\Delta y}
v(\tau,n-\Delta y)
+
\frac{1-\theta}{\Delta y}
v(\tau^{\prime},n-\Delta y),
\end{eqnarray}
which requires $\varepsilon_{n}\frac{\Delta y}{\Delta \tau} \ge 1-\theta$
for the numerical solution not to admit local maxima/minima.
Hereafter, we focus on the implicit scheme $\theta=1$ for simplicity.

Let $\{v_{l,m}\}_{l=0,1,\dots,L, m=0,1,\dots,M}$ for $\tau,n > 0$ and $L,M\in \mathbb{N}$ be the solution to
\begin{eqnarray}
\left\{ \begin{array}{ll}
\left(\frac{1}{\Delta \tau}+\left(\frac{\sigma_{m}}{\Delta y}\right)^{2}\right)v_{l,m} 
= \frac{1}{\Delta \tau}v_{l-1,m} 
+ \frac{1}{2}\left(\frac{\sigma_{m}}{\Delta y}\right)^{2}v_{l,m-1} 
+ \frac{1}{2}\left(\frac{\sigma_{m}}{\Delta y}\right)^{2}v_{l,m+1}, & l \neq 0, m \neq 0,M, \\
v_{l,0}=v_{l-1,0},\;
v_{l,M}=q_{n}v_{l-1,M}+(1-q_{n})v_{l,M-1}, & l=1,2,\dots,L, \\
v_{0,m}=h(m\Delta y),& m=0,1,\dots,M, \label{eq:system}
\end{array} \right.
\end{eqnarray}
where $\Delta \tau = \tau/L, \Delta y = n/M$, $\sigma_{m}=\sigma(m\Delta y)$,
and $q_{n} \in (0,1)$.
The configurations
$q_{n}=\frac{\varepsilon_{n}\frac{\Delta y}{\Delta \tau}}{1+\varepsilon_{n}\frac{\Delta y}{\Delta \tau}}$
and $q_{n}=0$
correspond to our procedure and that of \citet{EkstromLotstedtSydowTysk2011}, respectively.
Let $A$ be an $M+1$-dimensional square matrix such that 
the first two conditions in (\ref{eq:system}) are expressed as 
\begin{eqnarray}
A
\begin{pmatrix}v_{l+1,0}  \\ v_{l+1,1} \\ \vdots \\ v_{l+1,M} \end{pmatrix}
=
\begin{pmatrix}v_{l,0}  \\ v_{l,1} \\ \vdots \\ v_{l,M} \end{pmatrix}.
\end{eqnarray}
The matrix $A$ is invertible if $\Delta \tau$ is sufficiently small:
$P=(p_{i,j})_{i,j=0}^{M}:=A^{-1}$.

The following proposition shows that the solution to (\ref{eq:system}) satisfies certain desirable properties.
\begin{proposition}
\label{prop:discreteanalysis}
Assume that the matrix $A$ is invertible,
and let $v_{l,m}$ be a solution to (\ref{eq:system}) with $h(y)=y$.
Then, $v_{l,m}$ is non-increasing with respect to $l$
and non-decreasing and concave with respect to $m$.
\end{proposition}
\begin{proof}
Once we show that
$\{v_{1,m}\}_{m=0}^{M}$ is non-decreasing and concave
and $v_{l,m} \le v_{0,m}$ 
for a non-decreasing and concave function $h$ that is not identically $0$ with $h(0)=0$,
the conclusion follows by induction with respect to $l$,
because the existence of $A^{-1}$ ensures that $\{v_{1,m}\}_{m=0}^{M}$ is not identically $0$.

Suppose that $v_{1,M} \ge v_{0,M}$.
Let $m$ be the smallest number at which the maximum of $\{v_{1,m}\}_{m=0}^{M}$ is attained.
Then, $m \neq 0,M$
because the hypothesis yields
\begin{eqnarray}
0 < v_{0,M} \le v_{1,M} = q_{n}v_{0,M}+(1-q_{n})v_{1,M-1} \le v_{1,M-1}.
\end{eqnarray}
In addition, we obtain $v_{1,m} \ge v_{1,m+1},v_{1,m} > v_{1,m-1}$ 
and $v_{1,m} \ge v_{1,M} \ge v_{0,M} \ge v_{0,m}$,
which contradicts the fact that $v_{1,m}$ is an average of $v_{1,m\pm 1},v_{0,m}$.

Subsequently, suppose that the maximum of $\{\Delta_{1,m}\}_{m=0}^{M}$ is strictly positive, where $\Delta_{l,m}:=v_{l,m}-v_{0,m}$ for $l=0,1$.
Then, the maximum is attained at some $m \neq 0,M$.
Based on the concavity of $h$, we obtain
\begin{eqnarray}
\left(\frac{1}{\Delta \tau}+\left(\frac{\sigma_{m}}{\Delta y}\right)^{2}\right)\Delta_{1,m} 
&\le& \frac{1}{\Delta \tau}\Delta_{0,m} 
+ \frac{1}{2}\left(\frac{\sigma_{m}}{\Delta y}\right)^{2}\Delta_{1,m-1} 
+ \frac{1}{2}\left(\frac{\sigma_{m}}{\Delta y}\right)^{2}\Delta_{1,m+1} \nonumber\\
&=&
\frac{1}{2}\left(\frac{\sigma_{m}}{\Delta y}\right)^{2}\Delta_{1,m-1} 
+ \frac{1}{2}\left(\frac{\sigma_{m}}{\Delta y}\right)^{2}\Delta_{1,m+1} \nonumber\\
&\le& \frac{1}{2}\left(\frac{\sigma_{m}}{\Delta y}\right)^{2} \max \{\Delta_{1,m+1},\Delta_{1,m-1} \},
\end{eqnarray}
which is a contradiction.
We obtain $\Delta_{1,m} \le 0$, which yields the concavity of $v_{1,m}$
and 
\begin{eqnarray}
v_{1,m}-v_{1,m-1} \ge v_{1,M}-v_{1,M-1} = \frac{q_{n}}{1-q_{n}}(v_{0,M}-v_{1,M}) \ge 0.
\end{eqnarray}
\end{proof}

Proposition \ref{prop:mp} is a discrete maximum principle,
for which the assumptions are satisfied by the solution to (\ref{eq:system}):
\begin{eqnarray}
v_{l,M}-v_{l,M-1} = q_{n}(v_{l-1,M}-v_{l,M-1})
\ge q_{n}(v_{l,M}-v_{l,M-1}).
\end{eqnarray}
Then, $P$ is a transition probability matrix; that is, $p_{i,j} \ge 0$ and $\sum_{j=0}^{M}p_{i,j}=1$ for each $i=0,\dots,M$.
From Proposition \ref{prop:discreteanalysis},
the Markov process determined by $P$ is a supermartingale.
\begin{proposition}
\label{prop:mp}
Let $v_{l,m}$ be a solution to
\begin{eqnarray}
\left\{ \begin{array}{ll}
\left(\frac{1}{\Delta \tau}+\left(\frac{\sigma_{m}}{\Delta y}\right)^{2}\right)v_{l,m} 
\ge \frac{1}{\Delta \tau}v_{l-1,m} 
+ \frac{1}{2}\left(\frac{\sigma_{m}}{\Delta y}\right)^{2}v_{l,m-1} 
+ \frac{1}{2}\left(\frac{\sigma_{m}}{\Delta y}\right)^{2}v_{l,m+1}, & l \neq 0, m \neq 0,M, \\
v_{l,0} \ge 0,\;
v_{l,M} \ge v_{l,M-1}, & l=1,2,\dots,L, \\
v_{0,m} \ge 0,& m=0,1,\dots,M.
\end{array} \right.
\end{eqnarray}
Then, $v_{l,m} \ge 0$ for $l=0,1,2,\dots,L, m=0,1,\dots,M$.
\end{proposition}
\begin{proof}
Suppose that $v_{l,m} \ge 0$ for $m=0,1,\dots,M$ holds for some $l>0$,
which is the case for $l=0$ based on the initial condition.
If the minimum of $\{v_{l+1,m}\}_{m=0}^{M}$ is negative and attained at $m$,
$m \neq 0,M$
because of $v_{l+1,0} \ge 0$ and $v_{l+1,M} \ge v_{l+1,M-1}$.
However, $v_{l+1,m}$ is larger than the average of $v_{l,m},v_{l+1,m\pm 1}$,
which contradicts $v_{l,m} \ge 0 > v_{l+1,m}$ and $v_{l+1,m\pm 1} \ge v_{l+1,m}$.
\end{proof}

Finally, we compare our numerical solution with that of \citet{EkstromLotstedtSydowTysk2011}.
\begin{proposition}
\label{prop:neumann}
Let $v_{l,m}$ be the solution to (\ref{eq:system}) with $q_{n}=0$
and $v_{l,m}^{*}$ be that with $q_{n}=\frac{\varepsilon_{n}\frac{\Delta y}{\Delta \tau}}{1+\varepsilon_{n}\frac{\Delta y}{\Delta \tau}}$.
Then, we obtain $v_{l,m}^{*} \ge v_{l,m}$.
\end{proposition}
\begin{proof}
Proposition \ref{prop:mp} is applied to $\Delta_{l,m}:=v_{l,m}^{*}-v_{l,m}$.
\end{proof}

We study the derivatives of the solutions at $(\tau,y)=(0,n)$
with respect to $\tau$ as $\Delta \tau, \Delta y \longrightarrow 0$ in the case of $h(y)=y$.
This quantity in our procedure is bounded from below:
\begin{eqnarray}
\frac{v_{1,M}^{*}-v_{0,M}^{*}}{\Delta \tau}
= -\frac{1}{\varepsilon_{n}} \frac{v_{1,M}^{*}-v_{1,M-1}^{*}}{\Delta y}
\ge -\frac{1}{\varepsilon_{n}} \frac{v_{1,1}^{*}-v_{1,0}^{*}}{\Delta y}
\ge -\frac{1}{\varepsilon_{n}},\label{eq:slope}
\end{eqnarray}
which follows from
\begin{eqnarray}
\left(\frac{1}{\Delta \tau}+\left(\frac{\sigma_{1}}{\Delta y}\right)^{2}\right) (v_{1,1}^{*}-v_{1,0}^{*})
&=&
\frac{1}{\Delta \tau}v_{0,1}^{*}+\frac{1}{2}\left(\frac{\sigma_{1}}{\Delta y}\right)^{2}v_{1,2}^{*}
- \left(\frac{1}{\Delta \tau}+\frac{1}{2}\left(\frac{\sigma_{1}}{\Delta y}\right)^{2}\right)v_{1,0}^{*}
\nonumber\\
&\le& 
\frac{1}{\Delta \tau}\Delta y+\frac{1}{2}\left(\frac{\sigma_{1}}{\Delta y}\right)^{2}2 \Delta y,
\end{eqnarray}
and from the concavity of $v_{l,m}^{*}$ with respect to $m$.
For the solution of \citet{EkstromLotstedtSydowTysk2011},
\begin{eqnarray}
\frac{v_{1,M}-v_{0,M}}{\Delta \tau}
=\frac{v_{1,M-1}-v_{0,M}}{\Delta \tau}
\le \frac{v_{0,M-1}-v_{0,M}}{\Delta \tau}
= -\frac{\Delta y}{\Delta \tau}
\end{eqnarray}
can explode to $-\infty$.

\subsection{Numerical tests}
\label{sec:numericaltest}
We demonstrate that our procedure outperforms the others through numerical tests.
We use CEV models (\ref{eq:cev}) with two configurations, $(a,\nu)=(0.5,1)$ and $(0.5,1/3)$,
and compute the prices of two derivatives: one is a forward contract with maturity $T=1$ and the other is a derivative that pays $y^{p}$, $p=1+0.9/(2\nu)$ at time $T$.
In each test, the prices are computed using the five aforementioned methods as well as an analytical one, with the underlying prices varied.
We apply implicit finite difference methods, namely $\theta = 1$, on the finite grid $(0,T) \times (0,n)$ with $n = 40$ for the numerical procedures, except for that of \citet{Cetin2018}, as follows:
For each boundary condition on $y=n$, the numerical solutions to the Black--Scholes equation (\ref{eq:bse}) are computed on a uniform grid with $\Delta \tau = 0.01$ and $\Delta y = 0.05$.
For \citet{Cetin2018}, we compute $u$, which is the solution to (\ref{eq:cetin2}) with 
\begin{eqnarray}
T_{1/f}(x) = \frac{2\nu-1}{x}
\end{eqnarray}
on the finite grid $(0,T) \times (0,8)$ with $\Delta \tau=0.01$ and $\Delta x = 0.01$,
and obtain the prices using $v^{h}(\tau,f(x))=f(x)u(\tau,x)$ with $f(x)=(2\nu/a)^{2\nu}x^{-2\nu}$.

The results are reported in Tables \ref{table:forward1} and \ref{table:power1} for $\nu = 1$, and in Tables \ref{table:forward1_3} and \ref{table:power1_3} for $\nu=1/3$.
The numerical tests with $\nu=1$ demonstrate that the method of \citet{Cetin2018} and ours are the most accurate.
More precisely, the former is slightly better than ours.
However, the tests with $\nu=1/3$ show that 
the methods of \citet{Cetin2018} and \citet{SongYang2015} are inferior to the other three methods.

The spot volatility $\sigma(y)/y$ changes its behaviors according to $\nu \gtrless 1/2$, which corresponds to the sign of $T_{1/f}$.
Although it is increasing in both cases, it is concave in the case $\nu>1/2$ and convex in the case $\nu<1/2$.
In the latter case, the volatility increases rapidly with the underlying price, which results in almost constant prices of the derivatives once the underlying price exceeds a certain level.
The behaviors of the probability $P_{\cdot}^{1/f}[\tau < T_{0}]$ in the neighborhood of $0$ are also different.
The right-derivative at $0$ is divergent in the latter case.
This possibly makes it difficult for the method of \citet{Cetin2018} to compute the prices in the neighborhood of $0$.
The results are presented in Table \ref{table:prob}.

As mentioned previously, the Peclet condition (\ref{eq:peclet_cev}) for the method of \citet{Cetin2018} is violated in the case $\nu > 3/2$.
However, according to our tests, the accuracy and stability are similar to those of the case $\nu=1$.
The deterioration in accuracy in the method of \citet{Cetin2018} for $\nu = 1/3$ can be explained by
verifying whether $v^{*}(\tau,f(x)) = f(x) u^{*}(\tau,x)$ is decreasing with respect to $x$; that is,
the derivative $v^{*}(\tau,f(x))$ must be negative:
\begin{eqnarray}
f^{\prime}(x)u^{*}(\tau,x) + f(x) u_{x}^{*}(\tau,x)
= \left( \frac{-2\nu}{x}u^{*}(\tau,x) + u_{x}^{*}(\tau,x) \right)f(x)
< 0.
\end{eqnarray}
By discretizing the derivative $u_{x}^{*}$ and letting $x=\Delta x$, 
we obtain
\begin{eqnarray}
2\nu \frac{u^{*}(\tau,\Delta x)}{\Delta x} > u_{x}^{*}(\tau,\Delta x) \sim \frac{u^{*}(\tau,2\Delta x)}{2\Delta x}
\end{eqnarray}
and
\begin{eqnarray}
(4\nu-1) \frac{u^{*}(\tau,\Delta x)}{\Delta x}
> \frac{u^{*}(\tau,2\Delta x)-u^{*}(\tau,\Delta x)}{\Delta x}
\end{eqnarray}
when arranging the terms.
The ratios of the left- and right-hand sides approximate the left- and right-derivative of $u_{x}^{*}(\tau,\cdot)$ at $x = \Delta x$, respectively.
This condition fails if $\nu \le 1/4$ owing to the monotonicity of $u^{*}(\tau,\cdot)$,
and is likely to be violated if $4\nu - 1$ is positive but small.

\clearpage
\begin{table}[htbp]
\begin{center}
\caption{Price with $\nu = 1$ and payoff $h(y)=y$}
\label{table:forward1}
\begin{tabular}{l|rrrrrrrrrrrr} 
$y$  &  1.0 &  1.5 &  2.0 &  2.5 &  3.0 &  3.5 &  4.0 &  4.5 &  5.0 \\\hline
\citet{EkstromLotstedtSydowTysk2011}  &  1.00 &  1.49 &  1.94 &  2.34 &  2.70 &  3.01 &  3.28 &  3.51 &  3.72 \\
\citet{SongYang2015}  &  1.00 &  1.47 &  1.90 &  2.25 &  2.54 &  2.77 &  2.95 &  3.09 &  3.20 \\
\citet{Cetin2018}  &  1.00 &  1.49 &  1.96 &  2.40 &  2.79 &  3.15 &  3.47 &  3.75 &  4.01 \\
\citet{tsuzuki2023pitmans}  &  1.00 &  1.48 &  1.93 &  2.32 &  2.65 &  2.94 &  3.18 &  3.39 &  3.56 \\
This study  &  1.00 &  1.49 &  1.96 &  2.39 &  2.77 &  3.12 &  3.44 &  3.72 &  3.97 \\\hline
Exact  &  1.00 &  1.49 &  1.96 &  2.40 &  2.79 &  3.14 &  3.46 &  3.74 &  3.99 \\
\end{tabular}
\end{center}

\begin{center}
\caption{Price with $\nu = 1$ and payoff $h(y)=y^{1+0.9/(2\nu)}$}
\label{table:power1}
\begin{tabular}{l|rrrrrrrrrrrr} 
$y$  &  1.0 &  1.5 &  2.0 &  2.5 &  3.0 &  3.5 &  4.0 &  4.5 &  5.0 \\\hline
\citet{EkstromLotstedtSydowTysk2011}  &  1.11 &  2.07 &  3.16 &  4.24 &  5.25 &  6.18 &  7.02 &  7.76 &  8.43 \\
\citet{SongYang2015}  &  1.10 &  2.03 &  3.00 &  3.89 &  4.66 &  5.30 &  5.83 &  6.26 &  6.61 \\
\citet{Cetin2018}  &  1.11 &  2.15 &  3.43 &  4.83 &  6.27 &  7.67 &  9.00 & 10.24 & 11.40 \\
\citet{tsuzuki2023pitmans}  &  1.10 &  2.06 &  3.10 &  4.10 &  5.02 &  5.83 &  6.54 &  7.16 &  7.69 \\
This study  &  1.11 &  2.14 &  3.40 &  4.78 &  6.18 &  7.54 &  8.84 & 10.07 & 11.21 \\\hline
Exact  &  1.11 &  2.15 &  3.44 &  4.85 &  6.28 &  7.67 &  8.99 & 10.22 & 11.36 \\
\end{tabular}
\end{center}

\begin{center}
\caption{Price with $\nu = 1/3$ and payoff $h(y)=y$}
\label{table:forward1_3}
\begin{tabular}{l|rrrrrrrrrrrr} 
$y$  &  1.0 &  1.5 &  2.0 &  2.5 &  3.0 &  3.5 &  4.0 &  4.5 &  5.0 \\\hline
\citet{EkstromLotstedtSydowTysk2011}  &  0.89 &  1.01 &  1.05 &  1.07 &  1.07 &  1.08 &  1.08 &  1.08 &  1.08 \\
\citet{SongYang2015}  &  0.88 &  1.00 &  1.02 &  1.02 &  1.01 &  1.00 &  0.99 &  0.98 &  0.96 \\
\citet{Cetin2018}  &  0.89 &  1.01 &  1.04 &  1.05 &  1.05 &  1.05 &  1.05 &  1.05 &  1.04 \\
\citet{tsuzuki2023pitmans}  &  0.89 &  1.01 &  1.05 &  1.07 &  1.07 &  1.07 &  1.08 &  1.08 &  1.08 \\
This study  &  0.89 &  1.01 &  1.05 &  1.07 &  1.07 &  1.08 &  1.08 &  1.08 &  1.08 \\\hline
Exact  &  0.89 &  1.01 &  1.05 &  1.06 &  1.07 &  1.07 &  1.07 &  1.07 &  1.07 \\
\end{tabular}
\end{center}

\begin{center}
\caption{Price with $\nu = 1/3$ and payoff $h(y)=y^{1+0.9/(2\nu)}$}
\label{table:power1_3}
\begin{tabular}{l|rrrrrrrrrrrr} 
$y$  &  1.0 &  1.5 &  2.0 &  2.5 &  3.0 &  3.5 &  4.0 &  4.5 &  5.0 \\\hline
\citet{EkstromLotstedtSydowTysk2011}  &  1.02 &  1.43 &  1.56 &  1.61 &  1.64 &  1.65 &  1.65 &  1.66 &  1.66 \\
\citet{SongYang2015}  &  1.00 &  1.37 &  1.47 &  1.50 &  1.50 &  1.49 &  1.47 &  1.45 &  1.44 \\
\citet{Cetin2018}  &  1.02 &  1.41 &  1.53 &  1.58 &  1.59 &  1.60 &  1.60 &  1.59 &  1.59 \\
\citet{tsuzuki2023pitmans}  &  1.02 &  1.42 &  1.55 &  1.60 &  1.62 &  1.63 &  1.64 &  1.64 &  1.64 \\
This study  &  1.02 &  1.43 &  1.56 &  1.62 &  1.64 &  1.65 &  1.66 &  1.66 &  1.67 \\\hline
Exact  &  1.02 &  1.42 &  1.55 &  1.60 &  1.62 &  1.63 &  1.64 &  1.64 &  1.64 \\
\end{tabular}
\end{center}

\begin{center}
\caption{Probability $P_{\cdot}^{1/f}[\tau < T_{0}]$ computed using $v^{*}(\tau,y)/y$ with $\nu = 1/3$}
\label{table:prob}
\begin{tabular}{l|rrrrrrrrr} 
$x$ & 0.01& 0.10& 0.20& 0.30& 0.40& 0.50& 0.60& 0.70& 0.80\\
$y$ & 26.100 & 5.6229 & 3.5422 & 2.7032 & 2.2314 & 1.9230 & 1.7029 & 1.5366 & 1.4057 \\\hline
\citet{EkstromLotstedtSydowTysk2011} & 0.0414 & 0.1921 & 0.3038 & 0.3956 & 0.4751 & 0.5452 & 0.6075 & 0.6629 & 0.7122 \\
\citet{SongYang2015} & 0.0147 & 0.1682 & 0.2822 & 0.3760 & 0.4573 & 0.5291 & 0.5931 & 0.6501 & 0.7008 \\
\citet{Cetin2018} & 0.0337 & 0.1847 & 0.2970 & 0.3894 & 0.4695 & 0.5402 & 0.6030 & 0.6589 & 0.7086 \\
\citet{tsuzuki2023pitmans} & 0.0410 & 0.1917 & 0.3034 & 0.3952 & 0.4747 & 0.5448 & 0.6072 & 0.6626 & 0.7119 \\
This study & 0.0414 & 0.1921 & 0.3038 & 0.3956 & 0.4751 & 0.5452 & 0.6075 & 0.6629 & 0.7122 \\\hline
Exact & 0.0413 & 0.1913 & 0.3025 & 0.3939 & 0.4731 & 0.5430 & 0.6052 & 0.6606 & 0.7099 \\
\end{tabular}
\end{center}
\end{table}

\end{document}